%% file: SARDONICS.tex
\documentclass[a4paper,10pt]{article}

\usepackage{amsmath, amssymb, amsfonts}

\usepackage{amsthm}

\usepackage{graphicx, subfig}
\usepackage{pstricks}
\usepackage{anysize}
\usepackage{algorithmic}

\usepackage[ruled]{algorithm2e}
\renewcommand{\algorithmcfname}{ALGORITHM}
\SetAlFnt{\small}
\SetAlCapFnt{\small}
\SetAlCapNameFnt{\small}
\SetAlCapHSkip{0pt}
\IncMargin{-\parindent}

\newcommand{\xbold}{\mathbf{x}}
\newcommand{\sbold}{\mathbf{s}}
\newcommand{\ybold}{\mathbf{y}}
\newcommand{\ubold}{\mathbf{u}}
\newcommand{\sigmabold}{\boldsymbol{\sigma}}

\newcommand{\Scal}{\mathcal{S}}

\newcommand{\Pcal}{\mathcal{P}}

\newcommand{\mbs}[1]{\ensuremath{\boldsymbol{#1}}}
\newcommand{\vtheta}{\boldsymbol{\theta}}

\newtheorem{definition}{Definition}
\newtheorem{proposition}{Proposition}

\newsavebox{\Figurebox}
\graphicspath{{Figures/}}

\begin{document}


\title{Self-Avoiding Random Dynamics on Integer Complex Systems}

\author{
{\sc Firas Hamze}\thanks{D-wave Systems Inc.. Email: fhamze@dwavesys.com} \and {\sc Ziyu Wang}\thanks{University of British Columbia. Email: ziyuw@cs.ubc.ca} \and {\sc Nando de Freitas}\thanks{University of British Columbia. Email: nando@cs.ubc.ca} \\
}


\maketitle

\begin{abstract}
This paper introduces a new specialized algorithm for equilibrium Monte
Carlo sampling of binary-valued systems, which allows for large moves in the state space.
This is achieved by constructing self-avoiding walks (SAWs) in the state space. As a consequence, many bits are flipped in a single MCMC step. We name the algorithm
SARDONICS, an acronym for Self-Avoiding Random Dynamics on
  Integer Complex Systems. The algorithm has several free parameters, but we show that Bayesian optimization can be used to automatically tune them. SARDONICS performs remarkably well in a broad number of sampling tasks: toroidal ferromagnetic and frustrated Ising models, 3D Ising models, restricted Boltzmann machines and chimera graphs arising in the design of quantum computers. 
\end{abstract}



%


%

\section{Introduction}
\label{sec:intro}

Ising models, also known as Boltzmann machines, are ubiquitous models in physics, machine learning and spatial statistics
\cite{Ackley-85,Besag-74,Newman-99,Kindermann-80,Hopfield-84}. They have recently lead to a revolution in unsupervised learning known as deep learning, see for example \cite{Hinton-06a,Memisevic-09,Ranzato-10b,Lee-09,Marlin-10}.
There is also a remarkably large number of other statistical inference problems, where one can apply Rao-Blackwellization \cite{Hamze-04,Robert-04,Liu-01} to integrate out all continuous variables and end up with a discrete distribution. Examples include topic modeling and Dirichlet processes \cite{Blei-03}, Bayesian variable selection \cite{Tham-02}, mixture models \cite{Liu-03} and multiple instance learning \cite{Kueck-05}. Thus, if we had effective ways of sampling from discrete distributions, we would solve a great many statistical problems. Moreover, since inference in Ising models can be reduced to max-SAT and counting-SAT problems \cite{Barahona-82,Welsh-99,Bian-11}, efficient Monte Carlo inference algorithms for Ising models would be applicable to a vast domain of computationally challenging problems, including constraint satisfaction and molecular simulation.

Many samplers have been introduced to make large moves in continuous state spaces. Notable examples are the Hamiltonian and Riemann Monte Carlo algorithms \cite{Duane-87,Neal-10,Girolami-11}. However, there has been little comparable effort when dealing with general discrete state spaces. One the most popular algorithms in this domain is the Swendsen-Wang algorithm \cite{Swendsen-87}. This algorithm, as shown here, works well for sparse planar lattices, but not for densely connected graphical models. For the latter problems, acceptance rates to make large moves can be very small. For example, as pointed out in \cite{Munoz-03} Ising models with Metropolis
dynamics can require $10^{15}$ trials to leave a metastable state at low temperatures, and such a simulation would take $10^{10}$ minutes. For some of these models, it is however possible to compute the rejection probability of the next move. This leads to more efficient algorithms that always accept the next move \cite{Munoz-03,Hamze-07}. The problem with this is that at the next iteration the most favorable move is often to go back to the previous state. That is, these samplers may often get trapped in cycles.

To overcome this problem, this paper presents a specialized algorithm for equilibrium Monte
Carlo sampling of binary-valued systems, which allows for large moves in the state space.
This is achieved by constructing self-avoiding walks (SAWs) in the state space. As a consequence, many bits are flipped in a single MCMC step. 

We proposed a variant of this strategy
for constrained binary distributions in \cite{Hamze-10}. The method presented here applies to unconstrained systems. It has many advancements, but more free parameters than our previous version, thus making the sampler hard to tune. For this reason, we adopt a Bayesian optimization strategy \cite{Mockus-82,Brochu-09,Mahendran-11} to automatically tune these free parameters, thereby allowing for the construction of parameter policies that trade-off exploration and exploitation effectively.

Monte Carlo algorithms for generating SAWs for polymer simulation originated with \cite{Rosenbluth-55}.  More recently, biased SAW processes were adopted as proposal distributions for Metropolis algorithms in \cite{Siepmann-92}, where the method was named \emph{configurational bias Monte Carlo}. The crucial distinction between these seminal works and ours is that those authors were concerned with simulation of physical systems that inherently posess the self-avoidance property, namely molecules in some state space. More specifically, the physics of such systems dictated that no component of a molecule may occupy the same spatial location as another; the algorithms they devised took this constraint into account. In contrast, we are treating a different problem, that of sampling from a binary state-space, where a priori, no such requirement exists. Our construction involves the idea of imposing self-avoidance on sequences of generated states as a process to instantiate a rapidly-mixing Markov Chain on the binary state-space. It is therefore more related to the class of optimization algorithms known as Tabu Search \cite{Glover-89} than to classical polymer SAW simulation methods. To our knowledge, though, a correct equilibrium Monte Carlo algorithm using the Tabu-like idea of self-avoidance in state-space trajectories has not yet been proposed.


We should also point out that sequential Monte Carlo (SMC) methods, such as Hot Coupling and annealed importance sampling, have been proposed to sample from Boltzmann machines \cite{Hamze-05,Salakhutdinov-08}. Since such samplers often use an MCMC kernel as proposal distribution, the MCMC sampler proposed here could enhance those techniques. The same observation applies when considering other meta-MCMC strategies such as parallel tempering \cite{Geyer-91,Earl-05}, multicanonical Monte Carlo \cite{Berg-91,Gubernatis-00} and Wang-Landau \cite{Wang-01} sampling.

\section{Preliminaries}
\label{sec:prelims}

Consider a binary-valued system defined on the state space \( \Scal
\triangleq \{0, 1\}^{M} \), i.e. consisting of \( M \) variables each of
which can be \( 0 \) or \( 1 \). The probability of a state \( \xbold
= [x_{1}, \ldots x_{M}] \) is given by the Boltzmann distribution:
\begin{equation}
  \label{eq:BoltzDist}
  \pi(\xbold) = \frac{1}{Z(\beta)} e ^{-\beta E(\xbold) } 
\end{equation} 
where \( \beta \) is an \emph{inverse temperature.} An instance of
such a system is the ubiquitous \emph{Ising model} of statistical
physics, also called a \emph{Boltzmann machine} by the machine
learning community. Our aim in this paper is the generation of states
distributed according to a Boltzmann distribution specified by a
particular energy function \( E(\xbold) \). 

A standard procedure is to apply one of the \emph{local} Markov Chain
Monte Carlo (MCMC) methodologies such as the Metropolis algorithm or
the Gibbs (heat bath) sampler. As is well-known, these algorithms can
suffer from issues of poor equilibration (``mixing'') and trapping in
local minima at low temperatures. More sophisticated methods such as
{Parallel Tempering} \cite{Geyer-91} and the {flat-histogram}
algorithms (e.g. multicanonical \cite{Berg-91} or Wang-Landau \cite{Wang-01}
sampling) can often dramatically mitigate this problem, but they
usually still rely on local MCMC at some stage. The ideas presented in
this paper relate to MCMC sampling using larger changes of state than
those of local algorithms. They can be applied on their own or in
conjunction with the previously-mentioned advanced methods. In this
paper, we focus on the possible advantage of our algorithms over local methods.

Given a particular state \( \xbold \in \Scal \), we denote by \(
\Scal_{n}( \xbold ) \) the set of all states at Hamming distance \( n
\) from \( \xbold \). For example if \( M=3 \) and \( \xbold =
[1,1,1] \), then \( \Scal_{0}( \xbold ) = \{ [1,1,1] \} \), \(
\Scal_{1}(\xbold) = \{ [0,1,1], [1,0,1], [1,1,0] \} \), etc. Clearly
\( | \Scal_{n}( \xbold ) | = \binom{M}{n} \).

We define the set of \emph{bits} in two states \( \xbold, \ybold \)
that agree with each other: let \( \Pcal( \xbold, \ybold ) = \{ i |
x_{i} = y_{i} \} \). For instance if \( \xbold = [0,1,0,1] \) and \(
\ybold = [0,0,0,1] \), then \( \Pcal(\xbold, \ybold ) = \{ 1,3,4 \}
\). Clearly, \( \Pcal(\xbold, \ybold ) = \Pcal(\ybold, \xbold) \)
and \( \Pcal(\xbold, \xbold) = \{ 1, 2, \ldots M \} \).

Another useful definition is that of the \emph{flip operator}, which
simply inverts bit \( i \) in a state, \( F(\xbold, i) \triangleq (
x_{1}, \ldots, \bar{x_{i}}, \ldots, x_{M} )\), and its extension
that acts on a sequence of indices, i.e. \( F( \xbold, i_{1}, i_{2},
\ldots, i_{k} ) \triangleq F( F( \ldots F( \xbold, i_{1}), i_{2}),
\ldots, i_{k} )\).

For illustration, we describe a simple Metropolis
algorithm in this framework. Consider a state \( \xbold \in \Scal \);
a new state \( \xbold' \in \Scal_{1}(\xbold) \) is generated by
flipping a random bit of \( \xbold \), i.e. \(  \xbold' = F(\xbold,i)
\) for \( i \) chosen uniformly from \( \{1,\ldots, M \} \), and
accepting the new state with probability:
\begin{equation}
\label{eq:MetropDumb}
\alpha = \min(1, e^{-\beta( E(\xbold')-E(\xbold) ) } ) 
\end{equation} 
If all the \emph{single-flip neighbors} of \( \xbold \), i.e. the states
resulting from \( \xbold \) via a single bit perturbation, are of
higher energy than \( E(\xbold) \), the acceptance rate will be small
at low temperatures.

\section{SARDONICS}
\label{sec:sardonics}

In contrast to the traditional single-flip MCMC algorithms, the
elementary unit of our algorithm is a type of move that allows for
large changes of state \emph{and} tends to propose them such that they
are energetically favourable. We begin by describing the move operator
and its incorporation into the Monte Carlo algorithm we call
SARDONICS, an acronym for \emph{Self-Avoiding Random Dynamics on
  Integer Complex Systems.} This algorithm will then be shown to
satisfy the theoretical requirements that ensure correct asymptotic
sampling. The move procedure aims to force exploration away from the
current state. In that sense it has a similar spirit to the \emph{Tabu
  Search} \cite{Glover-89} optimization heuristic, but the aim of
SARDONICS is equilibrium sampling, a problem more general than (and in
many situations at least as challenging as) minimization.

Suppose that we have a state \( \xbold_{0} \) on \( \Scal \). It is
possible to envision taking a special type of biased \emph{self-avoiding
  walk} (SAW) of length \( k \) \emph{in the state space}, in other
words a sequence of states such that no state recurs in the sequence.
In this type of SAW, states on sets \( \{ \Scal_{1}( \xbold_{0} )
\ldots \Scal_{k}( \xbold_{0} ) \} \) are visited consecutively. To
generate this sequence of states \( (\ubold_{1}, \ldots, \ubold_{k} )
\), at step \( i \) of the procedure a bit is chosen to be flipped
from those that have not yet flipped relative to \( \xbold_{0} \).
Specifically, an element \( \sigma_{i} \) is selected from \( \Pcal(
\xbold_{0}, \ubold_{i-1} ) \), with \( \ubold_{0} \triangleq
\xbold_{0} \), to yield state \( \ubold_{i} = F( \ubold_{i-1},
\sigma_{i} ) \), or equivalently, \( \ubold_{i} = F( \xbold_{0},
\sigma_{1}, \ldots, \sigma_{i}) \). From \( \ubold_{i-1} \in
\Scal_{i-1}(\xbold_{0}) \), the set of states on \( \Scal_{i} \) that
can result from flipping a bit in \( \Pcal( \xbold_{0}, \ubold_{i-1} )
\) are called the \emph{single-flip neighbours} of \( \ubold_{i-1} \)
on \( \Scal_{i}( \xbold_{0} ) \).
The set of bits that can flip with nonzero probability are called the
\emph{allowable moves} at each step. A diagrammatic depiction of the
SAW is shown in Figure \ref{fig:SAW}.

The elements \( \{\sigma_{i}\} \) are sampled in an
\emph{energy-biased} manner as follows:

\begin{equation}
  \label{eq:fsawstep}
  f( \sigma_{i} = l | \sigma_{i-1}, \ldots, \sigma_{1}, \xbold_{0} ) = \left\{
    \begin{array}{ll}
      \frac{e^{-\gamma
          E( F(\ubold_{i-1}, l) ) }}{\sum_{ j \in \Pcal(\xbold_{0},\ubold_{i-1}) }e^{-\gamma E( F(\ubold_{i-1},j)
          )}} & 

      \begin{array}{ll} l \in \Pcal(\xbold_{0},\ubold_{i-1})
        \textrm{ and } \\ \ubold_{i-1} =
      F(\xbold_{0}, \sigma_{1}, \ldots, \sigma_{i-1}) \end{array} \\ \\
      0 & \textrm{otherwise}
    \end{array}
  \right.
\end{equation}
The larger the value that simulation parameter \( \gamma \) is
assigned, the more likely the proposal \( f \) is to sample the
lower-energy neighbors of \( \ubold_{i-1} \). Conversely if it is
zero, a neighbor on \( \Scal_{i}(\xbold_{0}) \) is selected completely
at random. In principle, the value of \( \gamma \) will be seen to be
arbitrary; indeed it can even be different at each step of the SAW.
We wil have more to say about the choice of \( \gamma \), as well as
the related issue of the SAW lengths, in Section \ref{sec:Adaptation}.

The motivation behind using such an energy-biased scheme is that when
proposing large changes of configuration, it may generate final states
that are ``typical'' of the system's target distribution. To make a
big state-space step, one may imagine uniformly perturbing a large
number of bits, but this is likely to yield states of high energy, and
an MCMC algorithm will be extremely unlikely to accept the move at low
temperatures.

\begin{figure}
\centering
\includegraphics[width=0.8\textwidth]{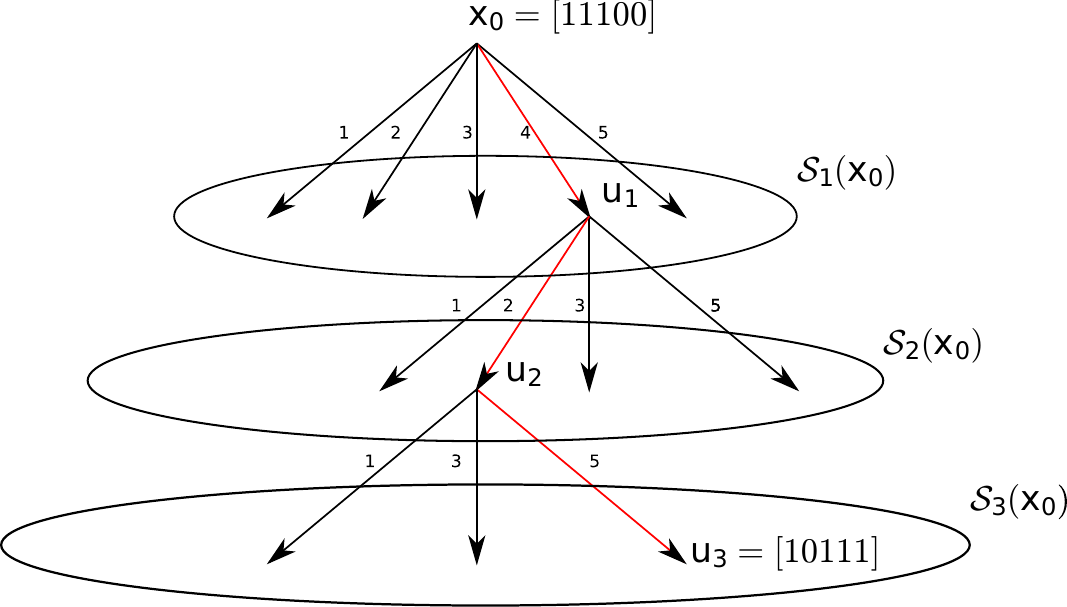}
\caption{ A visual illustration of the move process for a SAW of
 length 3. The arrows represent the allowable moves from a state at
 that step; the red arrow shows the actual move taken in this
 example. With the system at state \( \xbold_{0} \), the SAW
 begins. Bit 4 of \( \xbold_{0} \) has been sampled for flipping
 according to Equation \ref{eq:fsawstep} to yield state
 \(\ubold_{1}=[11110]\); the process is repeated until state \(
 \ubold_{3} \) on \( \Scal_{3}(\xbold_{0}) \) is reached. The
 sequence of states taken by the SAW is \( \sigmabold = [4,2,5] \) .}
\label{fig:SAW}
\end{figure}

At this point it can be seen why the term ``self-avoiding'' aptly
describes the processes. If we imagine the state-space to be a
high-dimensional lattice, with the individual states lying at the
vertices and edges linking states that are single-flip neighbors, a
self-avoiding walk on this graph is a sequence of states that induce a
connected, acyclic subgraph of the lattice. In the move procedure we
have described, a state can never occur twice at any stage within it
and so the process is obviously self-avoiding.

Note however that the construction imposes a stronger condition on the
state sequence; once a transition occurs from state \( \xbold \) to
state \( F(\xbold, i ) \), not only may state \( \xbold \) not appear
again, but neither may \emph{any} state \( \ybold \) with \( y_{i} = x_{i} \). It
seems natural to ask why not to use a less constrained SAW process,
namely one that avoids returns to \emph{individual} states and likely
more familiar to those experienced in molecular and polymer
simulations, without eliminating an entire dimension at each step.  In
our experience, trying to construct such a SAW as a proposal requires
excessive computational memory and time to yield good state-space
traversal. A local minimum ``basin'' of a combinatoric landscape can
potentially contain a massive number of states, and a process that
moves by explicitly avoiding particular states may be doomed to wander
within and visit a substantial portion of the basin prior to escaping.

Let \( \xbold_{1} \) be the \emph{final} state reached by the SAW, i.e. \(
\ubold_{k} \) for a length \( k \) walk. By multiplying the SAW
flipping probabilities, we can straightforwardly obtain the
probability of moving from state \( \xbold_{0} \) to \( \xbold_{1} \)
\emph{along the SAW} \( \sigmabold \), which we
call \( f( \xbold_{1}, \sigmabold | \xbold_{0} )\):
\begin{equation}
\label{eq:fsaw}
f( \xbold_{1}, \sigmabold | \xbold_{0} ) \triangleq 
\delta_{\xbold_{1}}[ F(\xbold_{0}, \sigmabold )] \prod_{i=1}^{k} f( \sigma_{i} | \ubold_{i-1} )
\end{equation}
The delta function simply enforces the fact that the final
\emph{state} \( \xbold_{1} \) must result from the sequence of flips in \(
\sigmabold \) from \( \xbold_{0} \). The set of \( \{\sigmabold\} \)
such that \( f(\xbold_{1}, \sigmabold|\xbold_{0}) >0 \) are termed the
\emph{allowable SAWs} between \( \xbold_{0} \) and \( \xbold_{1} \).

Ideally, to implement a Metropolis-Hastings (MH) algorithm using the
SAW proposal, we would like to evaluate the \emph{marginal}
probability of proposing \( \xbold_{1} \) from \( \xbold_{0} \), which
we call \( f( \xbold_{1} | \xbold_{0} ) \), so that the move would be
accepted with the usual MH ratio: 
\begin{equation}
\label{eq:margAlpha}
 \alpha_{m}(\xbold_{0}, \xbold_{1}) \triangleq \min\bigg(1, \frac{\pi(\xbold_{1})f(
  \xbold_{0} | \xbold_{1}
  )}{\pi(\xbold_{0})f( \xbold_{1} | \xbold_{0} )}\bigg)
 \end{equation}
 Unfortunately, for all but small values of the walk lengths \( k \),
 marginalization of the proposal is intractable due to the potentially
 massive number of allowable SAWs between the two states.

To assist in illustrating our solution to this, we recall that a
sufficient condition for a Markov transition kernel \( K \) to have
target \( \pi \) as its stationary distribution is \emph{detailed
  balance:}
\begin{equation}
\label{eq:DB}
\pi(\xbold_{0})K( \xbold_{1} | \xbold_{0} ) = \pi(\xbold_{1})K( \xbold_{0} | \xbold_{1} ) 
\end{equation}
One special case obtains if we used the marginalized proposal \(
f(\xbold_{1}|\xbold_{0}) \) followed by the MH accept rule,
\begin{equation} 
\label{eq:margK} K_m( \xbold_{1} | \xbold_{0})
\triangleq f(\xbold_{1}|\xbold_{0}) \alpha_{m}(\xbold_{0}, \xbold_{1})
\end{equation}
As we cannot compute \( f(\xbold_{1}|\xbold_{0}) \), we shall use a
kernel \( K(\xbold_{1}, \sigmabold |\xbold_{0}) \) defined on
the \emph{joint space} of SAWs and states, and show that with some
care, detailed balance (\ref{eq:DB}) can still \emph{hold marginally}.
It will be clear, though that this does \emph{not} mean that the
resultant marginal kernel \( K( \xbold_{1} | \xbold_{0} ) \) is the
same as that in (\ref{eq:margK}) obtained using MH acceptance on the
marginal \emph{proposal} .

Define the \emph{sequence reversal operator} \( R(\sigmabold) \) to
simply return a sequence consisting of the elements of \( \sigmabold
\) in reverse order; for example \( R([2,3,1,4]) = [4,1,3,2] \). One
can straightforwardly observe that each allowable SAW \( \sigmabold \)
from \( \xbold_{0} \) to \( \xbold_{1} \) can be uniquely mapped to
the allowable SAW \( R(\sigmabold) \) from \( \xbold_{1} \) to \(
\xbold_{0} \). For example in Figure \ref{fig:SAW}, the SAW \(
R(\sigmabold)=[5,2,4] \) can be seen to be allowable from \(
\xbold_{1} \) to \( \xbold_{0} \). Next, we have the following
somewhat more involved concept, a variant of which we introduced in
\cite{Hamze-10}:

\begin{definition} 
\label{def:pathDB}
Consider a Markov kernel \( K( \xbold_{1}, \sigmabold | \xbold_{0} )
\) whose support set coincides with that of (\ref{eq:fsaw}). We say
that \emph{pathwise detailed balance} holds if \[ \pi(\xbold_{0})K(
\xbold_{1}, \sigmabold | \xbold_{0} ) = \pi(\xbold_{1})K( \xbold_{0},
R(\sigmabold) | \xbold_{1} ) \] for all \( \sigmabold, \xbold_{0},
\xbold_{1} \).
\end{definition}

It turns out that pathwise detailed balance is a \emph{stronger
  condition} than marginal detailed balance. In other words, 

\begin{proposition}
\label{prop:pathmarg}
If the property in Definition \ref{def:pathDB} holds for a
transition kernel \( K\) of the type described there, then \(
\pi(\xbold_{0})K( \xbold_{1} | \xbold_{0} ) = \pi(\xbold_{1})K(
\xbold_{0} | \xbold_{1} ) \)
\end{proposition}
\begin{proof}
Suppose, for given \( \xbold_{0}, \xbold_{1} \), we summed both
sides of the equation enforcing pathwise detailed balance over all
allowable SAWs \( \{\sigmabold'\} \) from \( \xbold_{0} \) to
\( \xbold_{1} \), i.e. 
\[
\sum_{\sigmabold'} \pi(\xbold_{0})K( \xbold_{1}, \sigmabold'| \xbold_{0} ) =
\sum_{\sigmabold'}\pi(\xbold_{1})K( \xbold_{0}, R(\sigmabold') | \xbold_{1}
)
\]
The left-hand summation marginalizes the kernel over allowable
SAWs and hence results in \( \pi(\xbold_{0})K(\xbold_{1} | \xbold_{0}
) \). The observation above that each allowable SAW from \(
\xbold_{0} \) to \( \xbold_{1} \) can be reversed to yield an
allowable one from \( \xbold_{1} \) to \( \xbold_{0} \) implies that
the right-hand side is simply a re-ordered summation over all
allowable SAWs from \( \xbold_{1} \) to \( \xbold_{0} \), and can thus
be written as \( \pi(\xbold_{1})K(\xbold_{0} | \xbold_{1}) \).
\end{proof}

We are now ready to state the final form of the algorithm, which
can be seen to instantiate a Markov chain satisfying pathwise detailed
balance. After proposing \( (\xbold_{1}, \sigmabold) \) using
the SAW process, we accept the move with the ratio:
\begin{equation}
\label{eq:pathAlpha}
\alpha(\xbold_{0}, \xbold_{1}, \sigmabold ) \triangleq \min
\bigg(  1, \frac{\pi(\xbold_{1}) f( \xbold_{0},  R(\sigmabold) | \xbold_{1} )}{\pi(\xbold_{0}) f( \xbold_{1},
  \sigmabold | \xbold_{0})} \bigg)
\end{equation}
 
The computational complexity of evaluating this accept ratio is of the
same order as that required to sample the proposed SAWs/state; the
only additional operations required are those needed to evaluate the
reverse proposal appearing in the numerator, which are completely
analogous to those involved in calculating the forward proposal.

Let us take a closer look at the marginal transition kernel \( K(\xbold_{1} |
\xbold_{0} ) \). We can factor the joint \emph{proposal} into:
\[
f( \xbold_{1}, \sigmabold | \xbold_{0} ) = f( \xbold_{1} |
\xbold_{0} ) f( \sigmabold | \xbold_{0}, \xbold_{1} )
\]
Of course, if we are assuming that \( f(\xbold_{1} | \xbold_{0} ) \)
is intractable to evaluate, then the conditional \( f( \sigmabold |
\xbold_{0}, \xbold_{1} ) \) must be so as well, but it is useful to
consider. If we now summed both sides of the \emph{joint} probability
of moving from \( \xbold_{0} \) to \( \xbold_{1} \) over allowable
paths, we would observe:
\[
\sum_{\sigmabold' }\pi( \xbold_{0} ) K( \xbold_{1},
\sigmabold'| \xbold_{0} )  = \pi(\xbold_{0} ) f(
\xbold_{1} | \xbold_{0} )\sum_{\sigmabold' } f( 
\sigmabold'| \xbold_{0}, \xbold_{1} ) \alpha( \xbold_{0},
\xbold_{1}, \sigmabold' )
\]
The summation on the right-hand side is thus the
\emph{conditional expectation} of the accept rate given that we are
attempting to move from \( \xbold_{0} \) to \( \xbold_{1} \); we call it
\begin{equation}
\label{eq:effAlpha}
\alpha( \xbold_{0}, \xbold_{1} ) \triangleq \sum_{\sigmabold'} f( \sigmabold'| \xbold_{0}, \xbold_{1} )
\alpha( \xbold_{0}, \xbold_{1}, \sigmabold')
\end{equation}
and it defines an \emph{effective acceptance rate} between \(
\xbold_{0}\) and \( \xbold_{1} \) under the sampling regime described
since \( K( \xbold_{1} | \xbold_{0} )=f(\xbold_{1} | \xbold_{0} )
\alpha( \xbold_{0}, \xbold_{1} ) \). It is not difficult to show that \( \alpha(\xbold_{0},
\xbold_{1} ) \neq \alpha_{m}(\xbold_{0}, \xbold_{1} ) \), i.e. the
marginal accept rate for the joint proposal is not the same as the one
that results from using the marginalized proposal. In fact we can make a stronger statement:
\begin{proposition}
\label{prop:alphainequality}
For every pair of states \( (\xbold_{0}, \xbold_{1}) \), \(
\alpha(\xbold_{0},\xbold_{1}) \leq \alpha_m(\xbold_{0},\xbold_{1}) \)
\end{proposition}
\begin{proof}
  For conciseness, denote \( \frac{1}{\alpha_m(\xbold_{0},\xbold_{1})}
  =  \frac{\pi(\xbold_{0})f(\xbold_{1}|\xbold_{0})}{ \pi(\xbold_{1})f(\xbold_{0}|\xbold_{1})}   \) by \(
  C \).  Define the sets \( \mathcal{A} \triangleq \{ \sigmabold |
    \frac{f(R(\sigmabold)|\xbold_{1}, \xbold_{0})}
    {f(\sigmabold|\xbold_{0}, \xbold_{1} ) } \geq C \} \) and \( \mathcal{\bar{A}} \triangleq \{ \sigmabold |
    \frac{f(R(\sigmabold)|\xbold_{1}, \xbold_{0})}
    {f(\sigmabold|\xbold_{0}, \xbold_{1} ) } < C \} \).
    Then 
\begin{eqnarray}
     \alpha(\xbold_{0}, \xbold_{1}) &=& \sum_{\sigmabold'} f(
    \sigmabold'| \xbold_{0}, \xbold_{1} ) \min
\bigg(  1, \frac{\pi(\xbold_{1}) f( \xbold_{0}|\xbold_{1}) f(
  R(\sigmabold') | \xbold_{1},\xbold_{0} )}{\pi(\xbold_{0}) f(
  \xbold_{1} | \xbold_{0} )f(
  \sigmabold' | \xbold_{0}, \xbold_{1})} \bigg) \nonumber \\ &=& 
\sum_{\sigmabold' \in \mathcal{A}} f(
    \sigmabold'| \xbold_{0}, \xbold_{1} ) + 
\frac{\pi(\xbold_{1}) f( \xbold_{0}|\xbold_{1}) }{\pi(\xbold_{0}) f(
  \xbold_{1} | \xbold_{0} )}
\sum_{\sigmabold' \in
      \mathcal{\bar{A}}} f(
  R(\sigmabold') | \xbold_{1},\xbold_{0} ) \nonumber \\ &=&
\sum_{\sigmabold' \in \mathcal{A}} f(
    \sigmabold'| \xbold_{0}, \xbold_{1} ) + \frac{1}{C} \sum_{\sigmabold' \in
      \mathcal{\bar{A}}} f(
  R(\sigmabold') | \xbold_{1},\xbold_{0} ) \nonumber
\end{eqnarray}

But by definition, for \( \sigmabold \in \mathcal{A} \), \( f(
  \sigmabold | \xbold_{0},\xbold_{1} ) \leq \frac{1}{C} f(
  R(\sigmabold) | \xbold_{1},\xbold_{0} ) \).

Therefore,  
\begin{eqnarray}
 \alpha(\xbold_{0}, \xbold_{1})  &\leq& \frac{1}{C}\sum_{\sigmabold' \in \mathcal{A}} f(
    R(\sigmabold')| \xbold_{1}, \xbold_{0} ) + \frac{1}{C} \sum_{\sigmabold' \in
      \mathcal{\bar{A}}} f(
  R(\sigmabold') | \xbold_{1},\xbold_{0} ) \nonumber \\ &=& 
\frac{1}{C}\Big( \sum_{\sigmabold' \in \mathcal{A}} f(
    R(\sigmabold')| \xbold_{1}, \xbold_{0} ) + \sum_{\sigmabold' \in
      \mathcal{\bar{A}}} f(
  R(\sigmabold') | \xbold_{1},\xbold_{0} ) \Big) \nonumber \\ &=&
\frac{1}{C} 
= \alpha_m(\xbold_{0},\xbold_{1}) \nonumber
\end{eqnarray}
\end{proof}

There is another technical consideration to address. The reader may
have remarked that while detailed balance does indeed hold for our
algorithm, if the SAW length is constant at \( k > 1 \), then the
resulting Markov chain is no longer \emph{irreducible.} In other
words, not every \( \xbold \in \Scal \) is reachable with nonzero
probability regardless of the initial state. For example if \( k = 2
\) and the initial state \( \xbold_{0} = [0,0,0,0,0,0] \), then the
state \( \xbold=[0,0,0,0,0,1] \) can never be visited. Fortunately,
this is a rather easy issue to overcome; one possible strategy is to
randomly choose the SAW length prior to each step from a set that
ensures that the whole state space can eventually be visited. A
trivial example of such a set is any collection of integers that
include unity, i.e. such that single-flip moves are allowed. Another
is a set that includes consecutive integers, i.e. \( \{ k_{0}, k_{0}+1
\} \) for any \( k_{0} < M \). This latter choice could allow states
separated by a single bit to occur in two steps; in the example above,
if the set of lengths was \( \{3, 4 \} \) then we could have \(
[0,0,0,0,0,0] \rightarrow [0,0,1,1,1,1] \rightarrow [0,0,0,0,0,1] \).
While this shows how to enforce theoretical correctness of the
algorithm, in upcoming sections we will discuss the issue of practical
choice of the lengths in the set.

The experimental section discusses the practical matter of
efficiently sampling the SAW for systems with sparse connectivity, such
as the 2D and 3D Ising models.
 
\input{iteratedSimple}

\label{sec:Iterated}

\input{Mixture}
\label{sec:Mixture}

\input{adaptation}

\label{sec:Adaptation}

\input{experiments}

\label{sec:Experiments}

\section{Discussion}
\label{sec:Discussion}

The results indicate that the proposed sampler mixes very well in a broad range of models. However, as already pointed out, we believe that we need to consider alternatives to nonparametric methods in Bayesian optimization. Parametric methods would enable us to carry out infinite adaptation. This should solve the problems pointed out in the discussion of the results on the ferromagnetic 2D Ising model. Another important avenue of future work is to develop a multi-core implementation of SARDONICS as outlined in the previous section.

\section*{Acknowledgments}
\label{sec:ack}
The authors would like to thank Misha Denil and Helmut Katzgraber. The Swendsen-Wang code was modified from a slower version kindly made available by Iain Murray.
This research was supported by NSERC and a joint MITACS/D-Wave Systems grant. 


{
\bibliography{monteCarlo}
\bibliographystyle{plain}
}

\end{document}

%% file: iteratedSimple.tex
\subsection{Iterated SARDONICS}
\label{sec:itersardonics}

The SARDONICS algorithm presented in Section \ref{sec:sardonics} is
sufficient to work effectively on many types of system, for example
Ising models with ferromagnetic interactions. This section, however,
will detail a more advanced strategy for proposing states that uses
the state-space SAW of Section \ref{sec:sardonics} as its basic move.
The overall idea is to extend the trial process so that the search for
a state to propose can continue from the state resulting from a
\emph{concatenated series} of SAWs from \( \xbold_{0} \). 
The reader
familiar with combinatorial optimization heuristics will note the
philosophical resemblance to the class of algorithms termed
\emph{iterated local search} \cite{Hoos-04}, but will again bear in mind that
in the present work, we are interested in equilibrium sampling as
opposed to optimization.

We begin by noting that restriction to a \emph{single} SAW is unnecessary. We can
readily consider in principle an arbitrary concatenation of SAWs \( (
\sigmabold_{1}, \sigmabold_{2}, \ldots, \sigmabold_{N} ) \). A
straightforward extension of the SAW proposal is then to select some
number of iterations \( N \) (which need not be the same from one move
attempt to the next,) to generate \( \xbold_{1} \) from \( \xbold_{0}
\) by sampling from the concatenated proposal, defined to be
\begin{equation}
\begin{split}
\label{eq:elemISARDprop}
g(\xbold_{1}, \sigmabold_{1}, \ldots,
\sigmabold_{N} | \xbold_{0} ) \triangleq
f( \ybold^{1}, \sigma_{1} | \xbold_{0} ) f( \ybold^{2},
\sigma_{2} | \ybold^{1} ) \ldots \\ f( \ybold^{N},
\sigmabold_{N} | \ybold^{N-1} ) \delta_{\xbold_{1}}(\ybold^{N})
\end{split}
\end{equation}
and to accept the move with probability
\begin{equation}
\label{eq:ISARDMH}
\alpha( \xbold_{0}, \xbold_{1}, \sigmabold_{1} \ldots ,
\sigmabold_{N} ) = \min \bigg( 1, \frac{\pi(\xbold_{1})
  g( \xbold_{0}, R(\sigmabold_{N}), \ldots
  R(\sigmabold_{1}) | \xbold_{1} )}{\pi(\xbold_{0})
  g( \xbold_{1}, \sigmabold_{1}, \ldots, \sigmabold_{N} | \xbold_{0})} \bigg)
\end{equation}
In (\ref{eq:elemISARDprop}), the superscripted \( \{ \ybold^{i} \} \)
refer to the ``intermediate'' states on \( \Scal \) generated by the
flip sequences, and the functions \( f \) are the SAW proposals
discussed in Section \ref{sec:sardonics}. The proposed state \(
\xbold_{1} \) is identical to the final intermediate state \(
\ybold^{N} \); to avoid obscuring the notation with too many delta
functions we take it as implicit that the the proposal evaluates to
zero for intermediate states that do not follow from the flip
sequences \(\{\sigmabold_{i} \} \).

We refer to this as the \emph{iterated SARDONICS}
algorithm. Its potential merit over taking a single SAW is that it may
generate more distant states from \( \xbold_{0} \) that are
favorable. Unfortunately, a priori there is no guarantee that the
final state \( \xbold_{1} \) will be more worthy of acceptance than
the intermediate states visited in the process; it is computationally
wasteful to often propose long sequences of flips that end up
rejected, especially when potentially desirable states may have been
passed over. It is thus very important to choose the right value of $N$. For this reason, we 
will introduce Bayesian optimization techniques to adapt $N$ and other parameters
automatically in the following section.

%% file: Mixture.tex
\subsection{Mixture of SAWs}
\label{sec:mixsardonics}

A further addition we made to the basic algorithm was the
generalization of the proposal to a mixture of SAW processes.
Each segment of the iterated procedure introduced in Section
\ref{sec:Iterated} could in principle operate at a different level of
the biasing parameter \(\gamma\). A possible strategy one can envision
is to occasionally take a \emph{pair} of SAWs with the first at a
small value of \( \gamma \) and the next at a large value. The first
step encourages exploration away from the current state, and the
second a search for a new low-energy state. More specifically, for the
two SAWs \( (\sigmabold_{1}, \sigmabold_{2})\), we can have a proposal
of the form:
\[
f( \xbold_{1}, \sigmabold_{1}, \sigmabold_{2} | \xbold_{0} ) = f(
\ybold^{1}, \sigmabold_{1} | \xbold_{0}, \gamma_{L} ) f( \xbold_{1},
\sigmabold_{2} | \ybold^{1}, \gamma_{H} )
\]
where \( \gamma_{H}\) and \( \gamma_{L} \) are high and low
inverse temperature biases respectively. Unfortunately, such a method on its
own will likely result in a high rejection rate; the numerator of the
MH ratio enforcing detailed balance will be:
\[
f( \xbold_{0}, R(\sigmabold_{2}), R(\sigmabold_{1}) | \xbold_{1} ) = f(
\ybold^{1}, R(\sigmabold_{2}) | \xbold_{1}, \gamma_{L} ) f( \xbold_{0},
R(\sigmabold_{1}) | \ybold^{1}, \gamma_{H} )
\]
The probability of taking the reverse sequences will likely be very
small compared to those of the forward sequences. In particular, the
likelihood of descending at low-temperature biasing parameter along
the sequence \( R(\sigmabold_{1}) \), where \( \sigmabold_{1} \) was
generated with the high-temperature parameter, will be low.

Our simple approach to dealing with this is to define the proposal to
be a \emph{mixture} of three types of SAW processes.  The mixture
weights, \( P_{LL}, P_{HL}, P_{LH} \), with \( P_{LL} + P_{LH} +
P_{HL} =1 \), define, respectively, the frequencies of choosing a
proposal unit consisting of a \emph{pair} of SAWs sampled with \(
(\gamma_{L}, \gamma_{L}) \), \( (\gamma_{H}, \gamma_{L}) \), and \(
(\gamma_{L}, \gamma_{H}) \). The first proposal type encourages local
exploration; both SAWs are biased towards low-energy states. The
second one, as discussed, is desirable as it may help the sampler
escape from local minima. The last one may seem somewhat strange;
since it ends with sampling at \( \gamma_{H} \), it will tend to
generate states with high energy which will consequently be
rejected. The purpose of this proposal, however, is to assist in the
acceptance of the HL exploratory moves due to its presence in the
mixture proposal. Thus, \( P_{LH} \) is a parameter that must be
carefully tuned. If it is too large, it will generate too many moves
that will end up rejected due to their high energy; if too small, its
potential to help the HL moves be accepted will be diminished. The
mixture parameters are thus ideal candidates to explore the
effectiveness of adaptive strategies to tune MCMC.

%% file: adaptation.tex
\section{Adapting SARDORNICS with Bayesian Optimization}
\label{sec:bayesopt_mcmc}

SARDONICS has several free parameters: upper and lower bounds on the SAW length ($k_u$ and $k_l$ respectively), $\gamma_{H}$, $\gamma_{L}$, \( P_{LL}, P_{HL}, P_{LH} \) as explained in the previous section, and finally the number of concatenated SAWs $N$. We group these free parameters under the symbol $\vtheta = \{k_u, k_l, \gamma_{H},\gamma_{L},P_{LL},P_{HL},P_{LH}, N\}$. Each $\vtheta$ defines a stochastic policy, where the SAW length $k$ is chosen at random in the set $[k_l, k_u]$ and where the SAW processes are chosen according to the mixture probabilities \( P_{LL}, P_{HL}, P_{LH} \).
Tuning all these parameters by hand is an onerous task. Fortunately, this challenge can be surmounted using adaptation.
Stochastic approximation methods, at first sight, might appear to be good candidates for carrying out this adaptation. They have become increasingly popular in the subfield of adaptive MCMC \cite{Haario-01,Andrieu-01,Roberts-09,Vihola-10}. There are a few reasons, however, that force us to consider alternatives to stochastic approximation.

In our discrete domain, there are no obvious optimal acceptance rates that could be used to construct the objective function for adaptation. Instead, we choose to optimize $\vtheta$ so as to minimize the area under the auto-correlation function up to a specific lag. This objective was previously adopted
in \cite{Andrieu-01,Mahendran-11}. One might argue that is is a reasonable objective given that researchers and practitioners often use it to diagnose the convergence of MCMC algorithms.
However, the computation of gradient estimates for this objective is very involved and far from trivial \cite{Andrieu-01}. This motivates the introduction of a gradient-free optimization scheme known as Bayesian optimization \cite{Mockus-82,Brochu-09}. Bayesian optimization also has the advantage that it trades-off exploration and exploitation of the objective function. In contrast, gradient methods are designed to exploit locally and may, as a result, get trapped in unsatisfactory local optima. 

The proposed adaptive strategy consists of two phases: adaptation and sampling. In the adaptation phase Bayesian optimization is used to construct a randomized policy. In the sampling phase, a mixture of MCMC kernels selected according to the learned randomized policy is used to explore the target distribution. Experts in adaptive MCMC would have realized that there is no theoretical need for this two-phase procedure. Indeed, if the samplers are uniformly ergodic, which is the case in our discrete setting, and adaptation vanishes asymptotically, then ergodicity can still be established \cite{Roberts-09,Andrieu-01}. However, in our setting the complexity of the adaptation scheme increases with time. Specifically, Bayesian optimization, as we shall soon outline in detail, requires fitting a Gaussian process to $I$ points, where $I$ is the number of iterations of the adaptation procedure. In the worst case, this computation is $O(I^3)$. There are techniques based on conjugate gradients, fast multipole methods and low rank approximations to speed up this computation \cite{deFreitas-05,Halko-11}. However, none of these overcome the issue of increasing storage and computational needs. So, for pragmatic reasons, we restrict the number of adaptation steps. We will discuss the consequence of this choice in the experiments and come back to this issue in the concluding remarks.

The two phases of our adaptive strategy are discussed in more detail subsequently.

\subsection{Adaptation Phase}

Our objective function for adaptive MCMC is the area under the auto-correlation function up to a specific lag. This objective is intractable, but noisy
observations of its value can be obtained by running the Markov chain
for a few steps with a specific choice of parameters
$\vtheta_i$. Bayesian optimization can be used to
propose a new candidate $\vtheta_{i+1}$ by approximating the unknown
function using the entire history of noisy observations and a prior over this function.
 The prior distribution used in this paper
is a Gaussian process.

The noisy observations are used to obtain the predictive distribution
of the Gaussian process. An expected utility function derived in terms
of the sufficient statistics of the predictive distribution is
optimized to select the next parameter value $\vtheta_{i+1}$. The
overall procedure is shown in Algorithm \ref{alg:adaptive-mcmc}.  We
refer readers to \cite{Brochu-09} and \cite{Lizotte-08} for in-depth reviews
of Bayesian optimization.
\begin{algorithm}[t]
\caption{Adaptation phase of SARDONICS}
\label{alg:adaptive-mcmc}
\begin{algorithmic}[1]
{\footnotesize
   \FOR{$i=1,2,\dots, I $ }
      \STATE Run SARDONICS for $L$ steps with parameters $\vtheta_{i}$.
       \STATE Use the drawn samples to obtain a noisy evaluation of the objective function: $z_{i}=h(\vtheta_{i}) + \epsilon$.
       \STATE Augment the observation set $\mathcal{D}_{1:i} = \{\mathcal{D}_{1:i-1}, (\vtheta_{i}, z_{i})\}$.
       \STATE Update the GP's sufficient statistics.
       \STATE Find $\vtheta_{i+1}$ by optimizing an acquisition function:
       $\vtheta_{i+1} = \arg\max_{\vtheta} u(\vtheta |\mathcal{D}_{1:i})$.
  \ENDFOR
}
\end{algorithmic}
\end{algorithm}

The unknown objective function $h(\cdot)$ is assumed to be distributed
according to a Gaussian process with mean function $m(\cdot)$ and
covariance function $k(\cdot, \cdot)$:
\begin{align*}
h(\cdot) &\sim GP(m(\cdot), k(\cdot, \cdot)).
\end{align*}
We adopt a zero mean function $m(\cdot) = 0$ and an anisotropic
Gaussian covariance that is essentially the popular automatic relevance determination (ARD) kernel
\cite{Rasmussen-06}:
\begin{align*}
k(\vtheta_j, \vtheta_k) &=
\exp
\left(
-\frac{1}{2}
(\vtheta_j - \vtheta_k)^T
\textrm{diag}(\psi)^{-2}
(\vtheta_j - \vtheta_k)
\right)
\end{align*}
where $\psi \in \mathbb{R}^d$ is a vector of hyper-parameters. The Gaussian process is a surrogate model for the true objective, which involves intractable expectations with respect to the invariant distribution and the MCMC transition kernels.

We assume that the noise in the measurements is Gaussian: $z_i = h(\vtheta_i) + \epsilon$, $\epsilon \sim \mathcal{N}(0, \sigma^2_{\eta})$. It is possible to adopt other noise models \cite{Diggle-98}. Our Gaussian process emulator has hyper-parameters $\psi$ and $\sigma_{\eta}$. These hyper-parameters are typically computed by maximizing the likelihood \cite{Rasmussen-06}. In Bayesian optimization, 
 we can use Latin hypercube designs to select an initial set of parameters and then proceed to maximize the likelihood of the hyper-parameters iteratively \cite{Ye-98,Santner-03}. This is the approach followed in our experiments. However, a good alternative is to use either classical or Bayesian quadrature to integrate out the hyper-parameters \cite{Osborne-10,Rue-09}.

Let $\mathbf{z}_{1:i}\sim \mathcal{N}(0,\mathbf{K})$ be the $i$ noisy
observations of the objective function obtained from previous
iterations. 
(Note that the Markov chain is run for $L$ steps for each
discrete iteration $i$. The extra index to indicate this fact has been
made implicit to improve readability.) $\mathbf{z}_{1:i}$ and $h_{i+1}$
are jointly multivariate Gaussian:
\begin{align*}
\begin{bmatrix}
\mathbf{z}_{1:i} \\
h_{i+1}
\end{bmatrix}
&=
\mathcal{N}
\left(
\mathbf{0},
\begin{bmatrix}
\mathbf{K}+\sigma^2_{\eta}I & \mathbf{k}^T \\
\mathbf{k} & k(\vtheta, \vtheta)
\end{bmatrix}
\right),
\end{align*}
where
\begin{align*}
\mathbf{K} &=
\begin{bmatrix}
k(\vtheta_1, \vtheta_1) & \ldots & k(\vtheta_1, \vtheta_i) \\
\vdots & \ddots & \vdots \\
k(\vtheta_i, \vtheta_1) & \ldots & k(\vtheta_i, \vtheta_i)
\end{bmatrix}
\end{align*}
and $\mathbf{k} =
[
k(\vtheta, \vtheta_1) \; \ldots \;
k(\vtheta, \vtheta_i)]^T.$ 
All the above assumptions about the form of the prior distribution and
observation model are standard and less restrictive than they might
appear at first sight. The central assumption is that the objective function is
smooth. For objective functions with discontinuities, we need more sophisticated surrogate functions for the cost. We refer readers to \cite{Gramacy-04} and \cite{Brochu-09} for examples.

The predictive distribution for any value $\vtheta$ follows from the
Sherman-Morrison-Woodbury formula, where $\mathcal{D}_{1:i} =
(\vtheta_{1:i},\mathbf{z}_{1:i})$:
\begin{align*}
p(h_{i+1} | \mathcal{D}_{1:i},\vtheta) &=
\mathcal{N}
(\mu_i(\vtheta), \sigma^2_i(\vtheta)) \\
\mu_i(\vtheta) &=
\mathbf{k}^{T}(\mathbf{K + \sigma_{\eta}^2 I})^{-1}\mathbf{z}_{1:i} \\
\sigma^2_i(\vtheta) &=
k(\vtheta, \vtheta) -
\mathbf{k}^{T}(\mathbf{K + \sigma_{\eta}^2 I})^{-1}\mathbf{k}
\end{align*}

The next query point $\vtheta_{i+1}$ is chosen to maximize an
acquisition function, $u(\vtheta |\mathcal{D}_{1:i})$, that trades-off exploration (where
$\sigma^2_i(\vtheta)$ is large) and exploitation (where $\mu_i(\vtheta)$
is high). We adopt the expected improvement over the best candidate as
this acquisition function \cite{Mockus-82,Schonlau-98,Brochu-09}. This is a standard acquisition function for which asymptotic rates of convergence have been proved \cite{Bull-11}. However, we point out that there are a few other reasonable alternatives, such as Thompson sampling \cite{May-11} and upper confidence bounds (UCB) on regret \cite{Srinivas-10}. A comparison among these options as well as portfolio strategies to combine them appeared recently in \cite{Hoffman-11}. There are several good ways of optimizing the acquisition function, including the method of DIvided RECTangles (DIRECT) of \cite{Finkel-03} and many versions of the projected Newton methods of \cite{Bertsekas-87}. We found DIRECT to provide a very efficient solution in our domain. Note that optimizing the acquisition function is much easier than optimizing the original objective function. This is because the acquisition functions can be easily evaluated and differentiated.

\subsection{Sampling Phase}
\label{sec:sampling_phase}

The Bayesian optimization phase results in a Gaussian process on the $I$
noisy observations of the performance criterion $\mathbf{z}_{1:I}$,
taken at the corresponding locations in parameter space $\vtheta_{1:I}$.
This Gaussian process is
used to construct a discrete stochastic policy $p(\vtheta |
\mathbf{z}_{1:I})$ over the parameter space $\mbs{\Theta}$.
The Markov chain is run with
parameter settings randomly drawn from this policy at each step.

One can synthesize the policy $p(\vtheta|\mathbf{z}_{1:I})$ in several ways. The simplest is to use the mean of the GP to construct a distribution proportional to $\exp(\mu(\vtheta))$. This is the so-called Boltzmann policy.
We can sample $M$ parameter candidates $\vtheta_i$ according to this distribution. Our final sampler then consists of a mixture of $M$ transition kernels, where each kernel is parameterized by one of the $\vtheta_i$, $i=1,\ldots,M$. The distribution of the
samples generated in the sampling phase will approach the target
distribution $\pi(\cdot)$ as the number of iterations tends to
$\infty$ provided the kernels in this finite mixture are ergodic.

In high dimensions, one reasonable approach would be to use a multi-start optimizer to find maxima of the unnormalized Boltzmann policy and then perform local exploration of the modes with a simple Metropolis algorithm. This is a slight more sophisticated version of what is often referred to as the epsilon greedy policy.

The strategies discussed thus far do not take into account the uncertainty of the GP. A solution is to draw M functions according to the GP and then find he optimizer $\vtheta_i$ of each of these functions. This is the strategy followed in  \cite{May-11} for the case of contextual bandits. Although this strategy works well for low dimensions, it is not clear how it can be easily scaled.

%% file: experiments.tex
\section{Experiments}

Our experiments compare
SARDONICS to the popular Gibbs, block-Gibbs and Swendsen-Wang samplers. Several types
of binary-valued systems, all belonging to the general class of
undirected graphical model called the \emph{Ising model}, are used.
The energy of a binary state \( \sbold \), where \( s_{i} \in \{-1,1\}
\) is given by: \begin{equation}
  \label{eq:IsingHamiltonian}
  E(\sbold) = - \sum_{(i,j)} J_{ij} s_{i} s_{j} - \sum_{i}
  h_{i} s_{i} \nonumber
\end{equation}
(One can trivially map \( x_{i} \in \{0,1\}
\) to \( s_{i} \in \{-1,1\}\) and vice-versa.)
The interaction weights \( J_{ij} \) between variables \( i \)
and \( j \) are zero if they are topologically disconnected; positive
(also called ``ferromagnetic'') if they tend to have the same value;
and negative (``anti-ferromagnetic'') if they tend to have opposite
values. The presence of interactions of mixed sign can significantly
complicate Monte Carlo simulation due to the proliferation of local
minimas in the energy landscape. Interaction weights of different sign produce unsatisfiable constraints and cause the system to become ``frustrated''. 

\begin{table}[h]

\begin{center}
\begin{tabular}{lrrrrrrrr}

& \multicolumn{8}{c}{\textbf{Parameters of SARDONICS} ($\vtheta$)} \\
\cline{2-9} \\
&  $k_l$ & $k_u$ & $\gamma_{L}$ & $\gamma_{H}$ & $P_{LL}$ & $P_{HL}$ & $P_{LH}$ & $N$\\
\hline \\
\textbf{Range} & $\{1,\ldots,70\}$ & $\{2,\ldots,120\}$ & $[0.89,1.05]$ & $[0.9,1.15]$ & $[0,1]$& $[0,1]$& $[0,1]$ &$\{1,\ldots,5\}$
\end{tabular}
\caption{The ranges from which Bayesian Optimization chooses parameters for SARDONICS. The selection mechanism ensures that $k_u\geq k_l$ and $\gamma_H \geq \gamma_L$.}
\end{center}

\label{table:param_range}
\end{table}

The first set of experiments considers the behavior of Gibbs, SARDONICS and Swendsen-Wang on a ferromagnetic Ising model on a planar, regular grid of
size $60\times 60$. The model has connections between the nodes on
one boundary to the nodes on the other boundary for each dimension. As a result of these periodic boundaries, the model is a  square toroidal grid. Hence, each node has exactly four neighbors. In the first experiment, the interaction
weights, $J_{ij}$, are all 1 and the biases, $h_i$, are all 0. We test the three algorithms on this model at three different temperatures: $1$, $2.27$ and $5$.
 The value \( \beta = 1/2.27 \) corresponds to the so-called \emph{critical
  temperature}, where many interesting phenomena arise
\cite{Newman-99} but where simulation also becomes quite
difficult. 


\begin{figure}[t!]
\begin{tabular*}{10 cm}{c c}
  \raisebox{0.01cm}{\includegraphics[trim = 35mm 15mm 35mm 35mm, clip, width=0.4\textwidth]{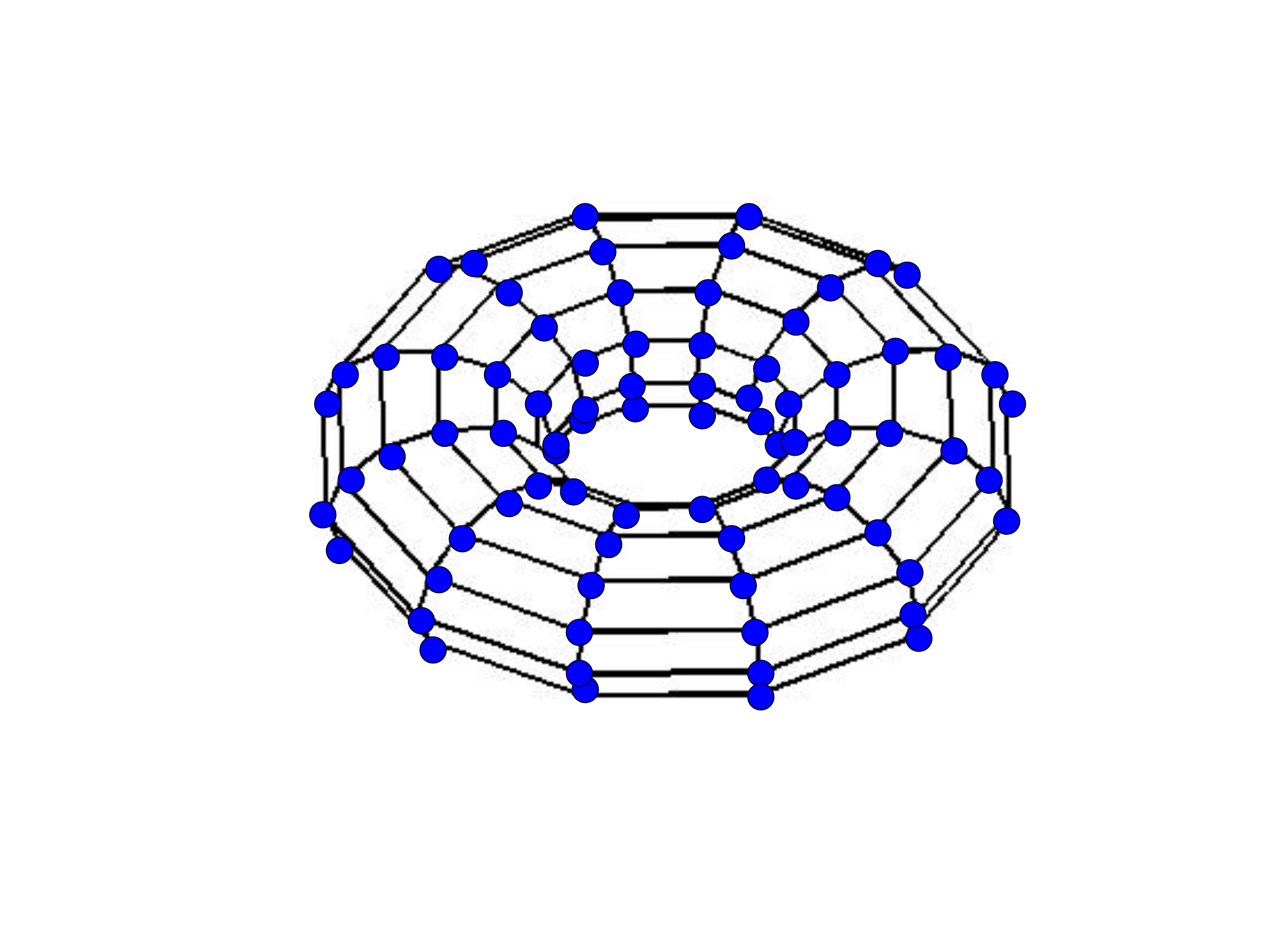}} &
  \includegraphics[width=0.5\textwidth]{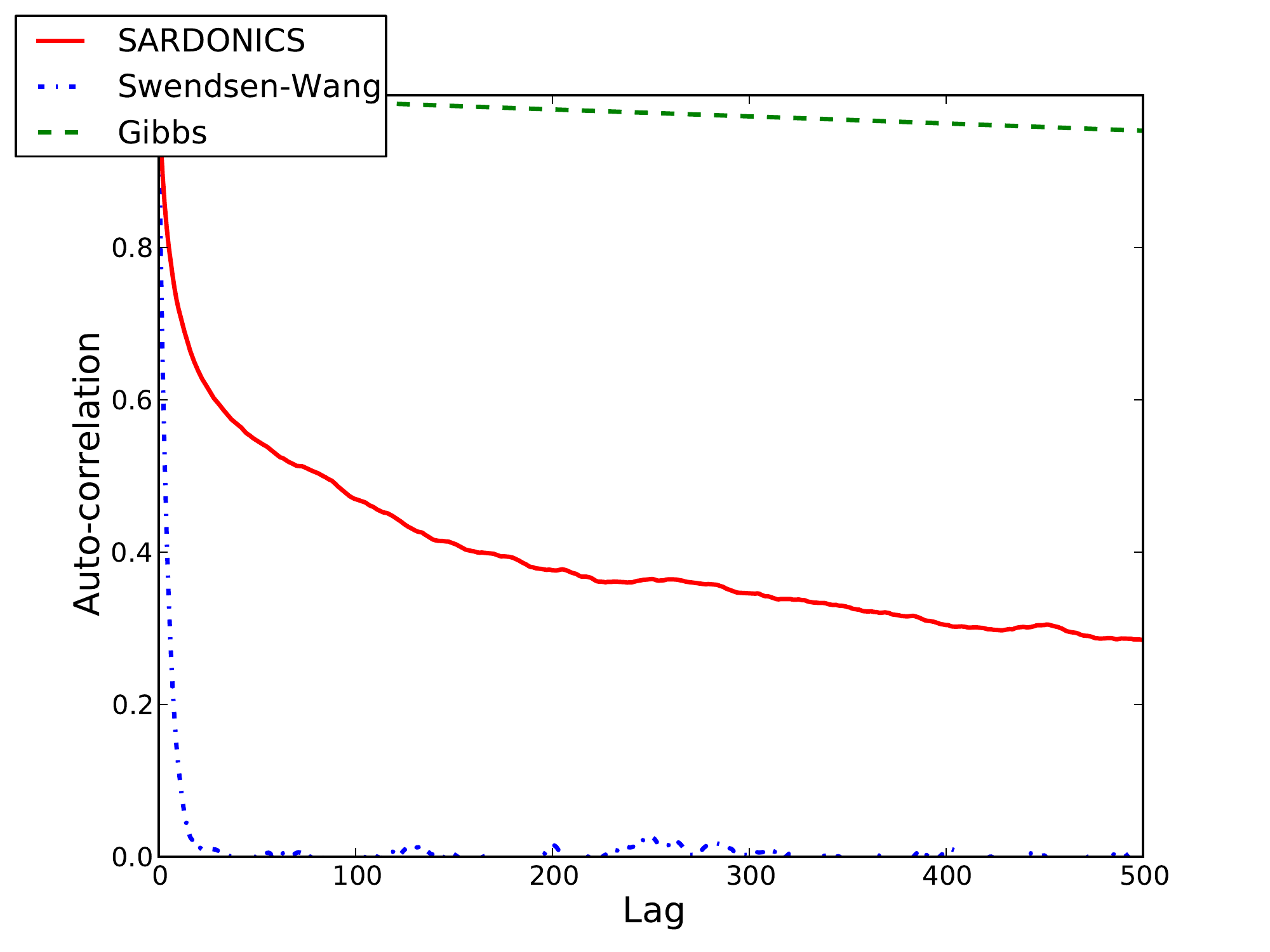} \\
    \includegraphics[width=0.45\textwidth]{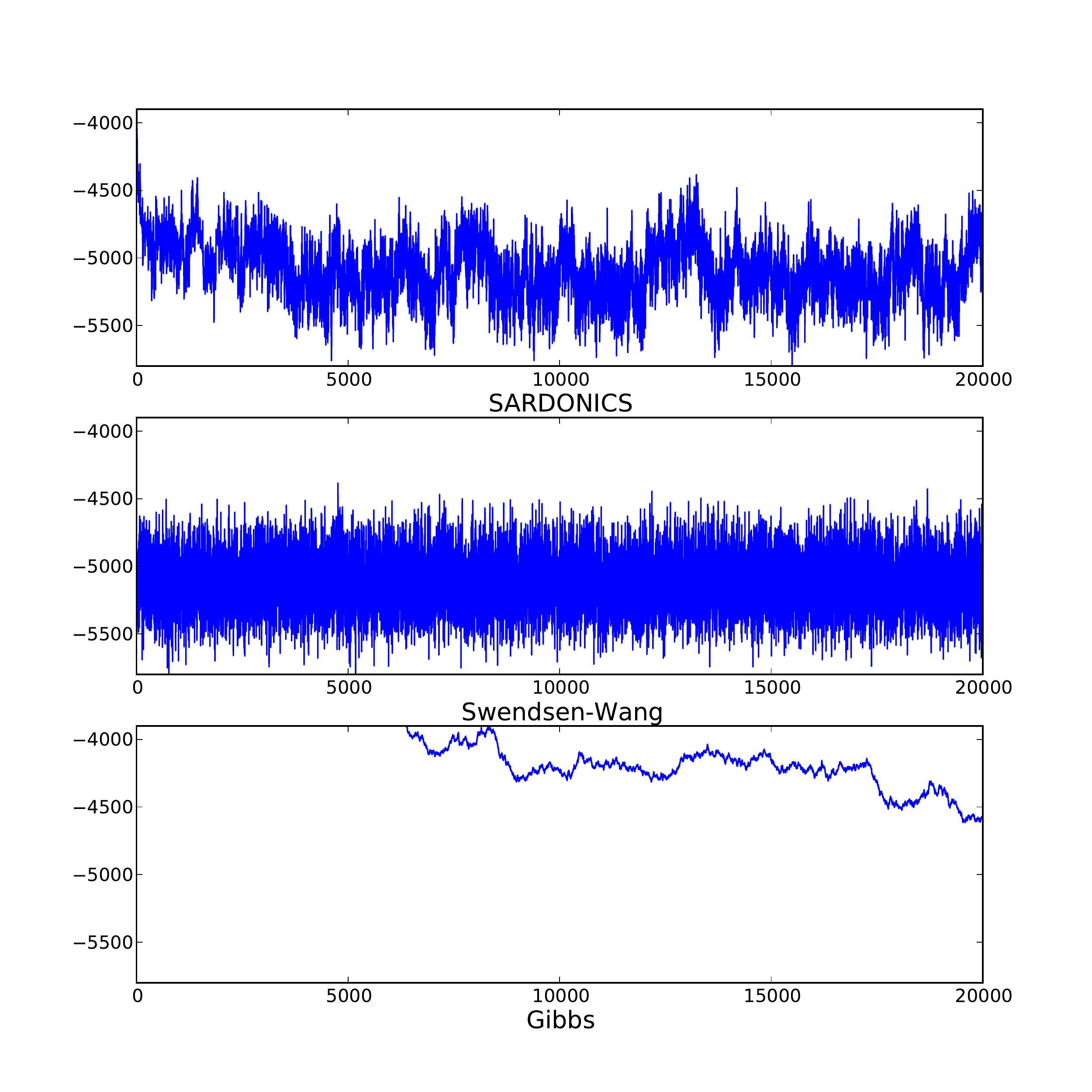} &
  \includegraphics[width=0.55\textwidth]{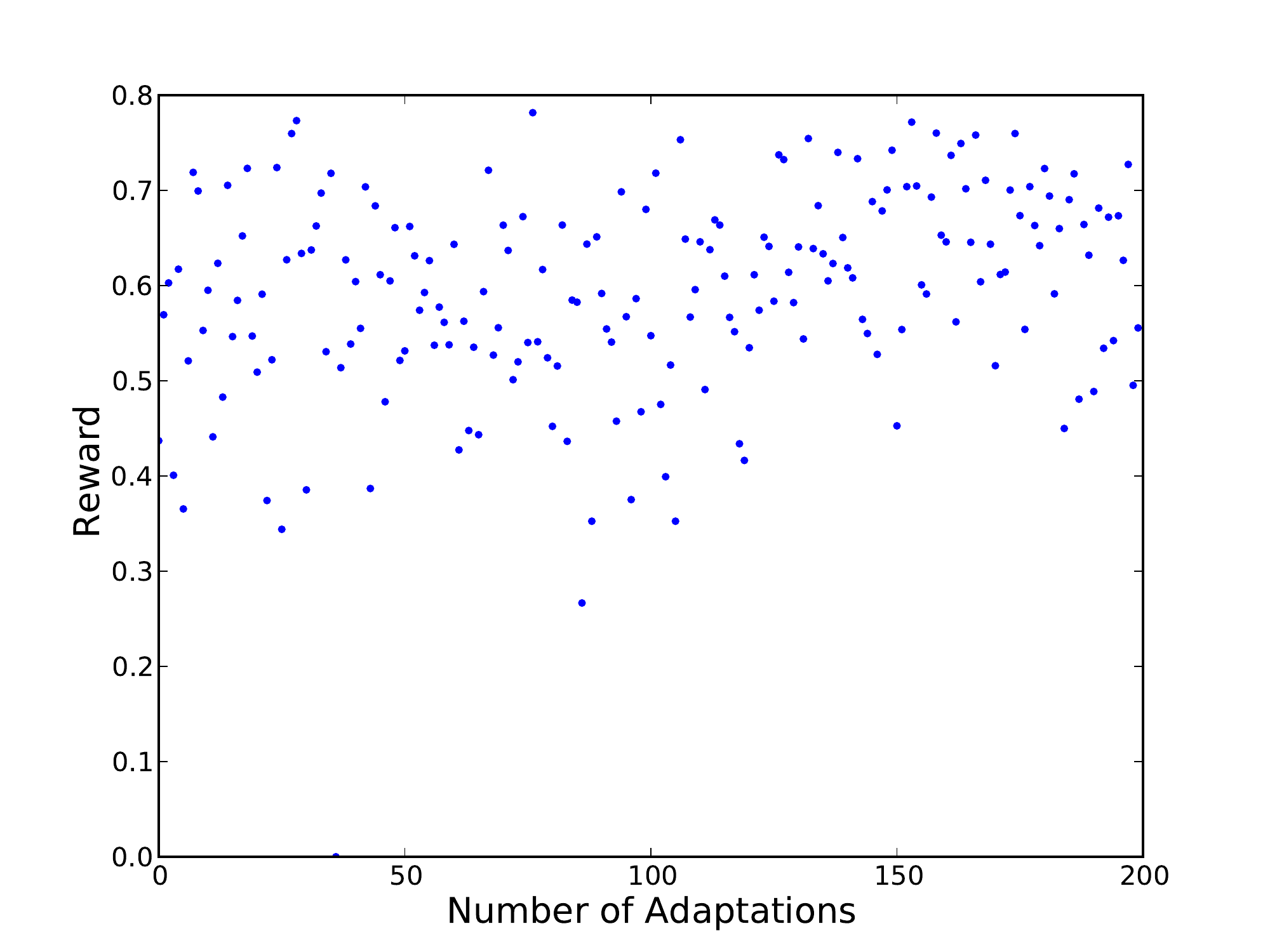}
\end{tabular*}
  \caption{The 2D Ferromagnetic model with periodic (toroidal) boundaries [top left], the auto-correlations of the samplers for $T=2.27$ (critical temperature) [top right], traces of every 5 out of the $10^5$ samples of the energy [bottom left] and rewards obtained by the Bayesian optimization algorithm during the adaptation phase [bottom right]. }
  \label{fig:acf1}
\end{figure}

The experimental protocol for this and subsequent models was the same: For
10 independent trials, run the competing samplers for a certain number
of iterations, storing the sequence of energies visited. Using each
trial's energy sequence, compute the \emph{auto-correlation function}
(ACF). Comparison of the algorithms consisted of analyzing the
energy ACF averaged over the trials. Without going into detail, a more
rapidly decreasing ACF is indicative of a faster-mixing Markov chain; see for example
\cite{Robert-04}. 
For all the models, each sampler is run for $10^5$ steps. For SARDONICS,
we use the first $2 \times 10^4$ iterations to adapt its hyper-parameters. For fairness
of comparison, we discard the first $2 \times 10^4$ samples from each sampler and compute
the ACF on the remaining $8 \times 10^4$ samples. 
For all our experiments, we use the ranges of parameters summarized in Table I to adapt SARDONICS.

\begin{figure}[t]
\begin{tabular*}{10 cm}{c c c}
    \includegraphics[width=0.5\textwidth]{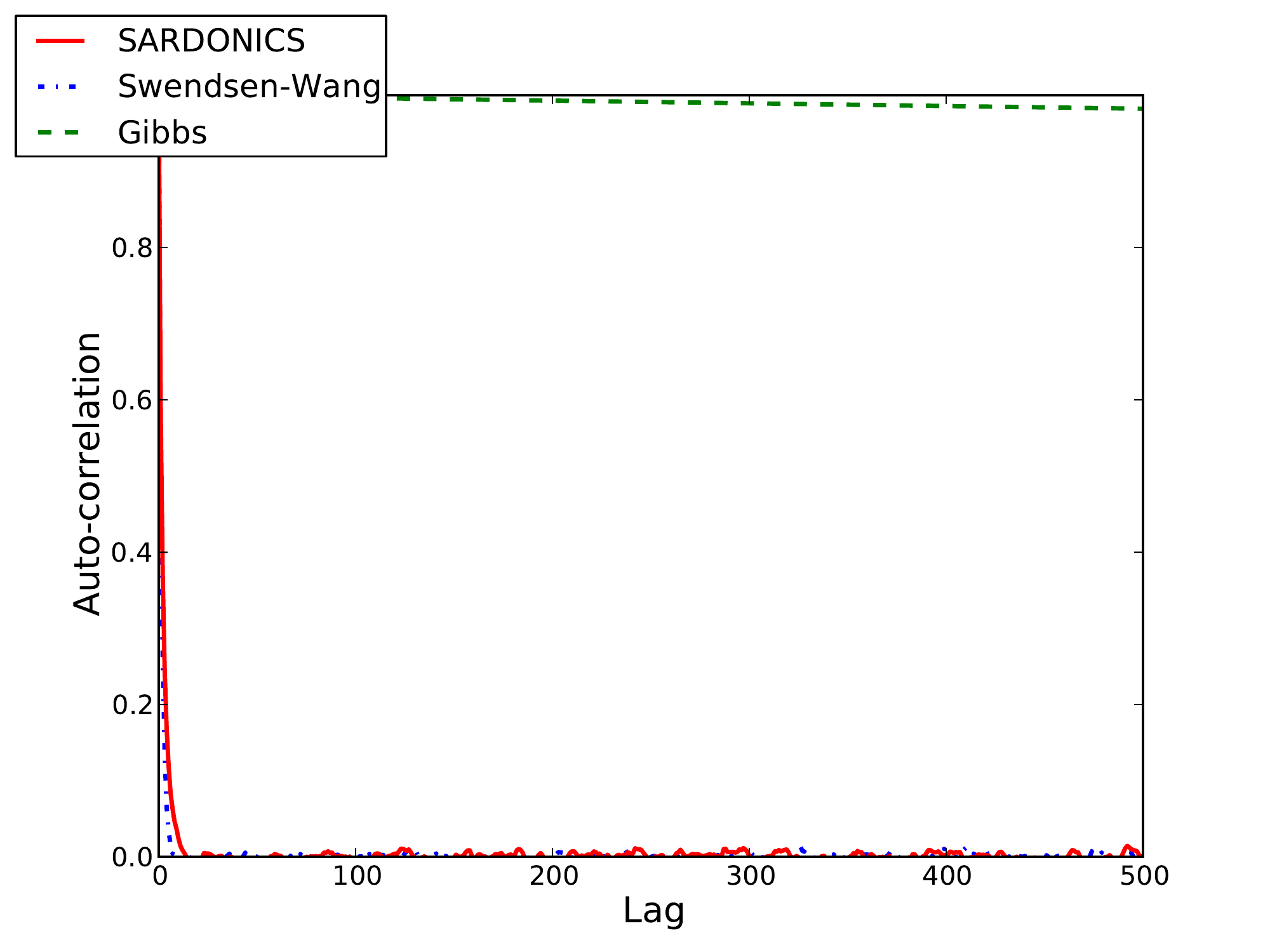} &
  \includegraphics[width=0.5\textwidth]{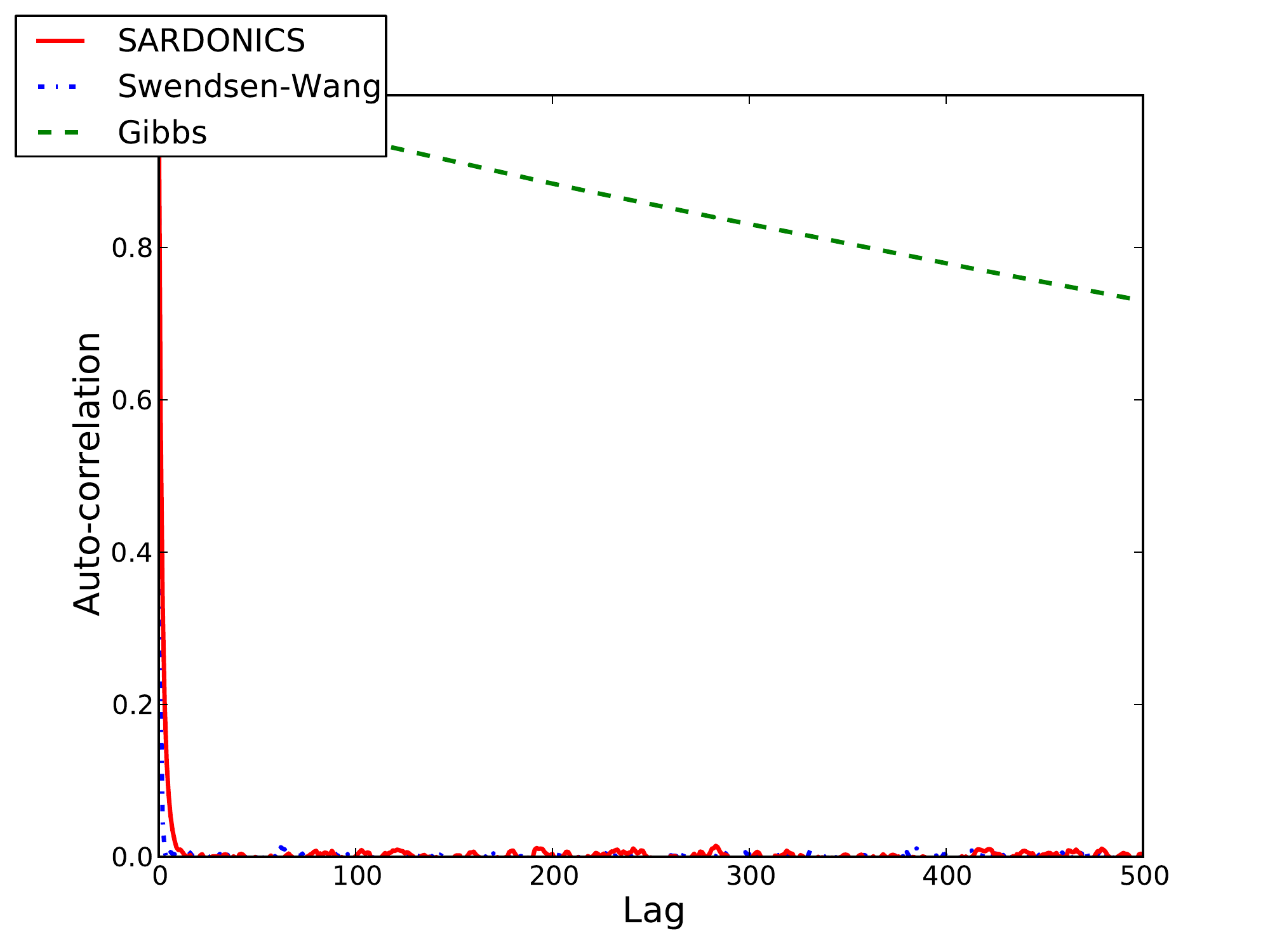} \\
\includegraphics[width=0.5\textwidth]{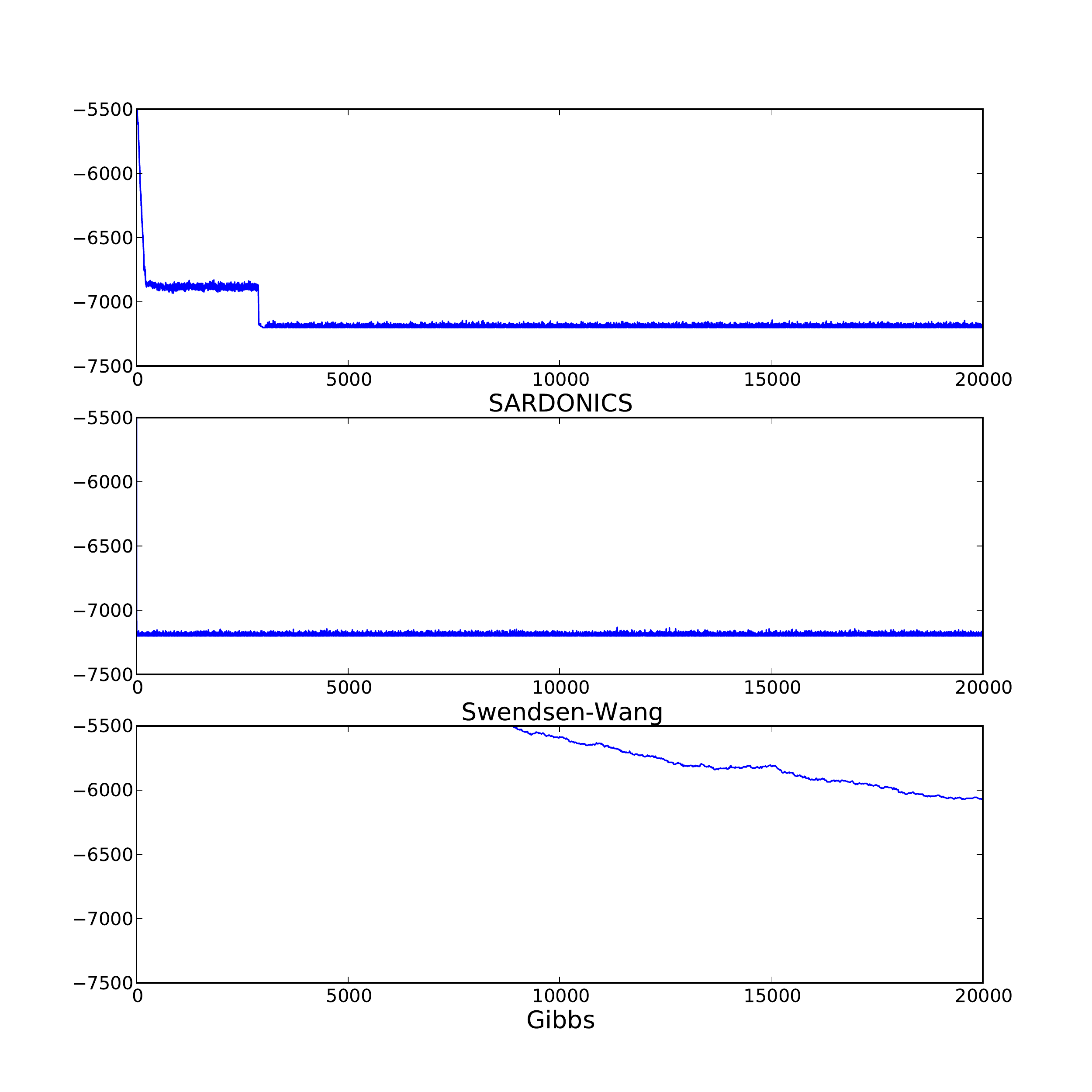} &
  \includegraphics[width=0.5\textwidth]{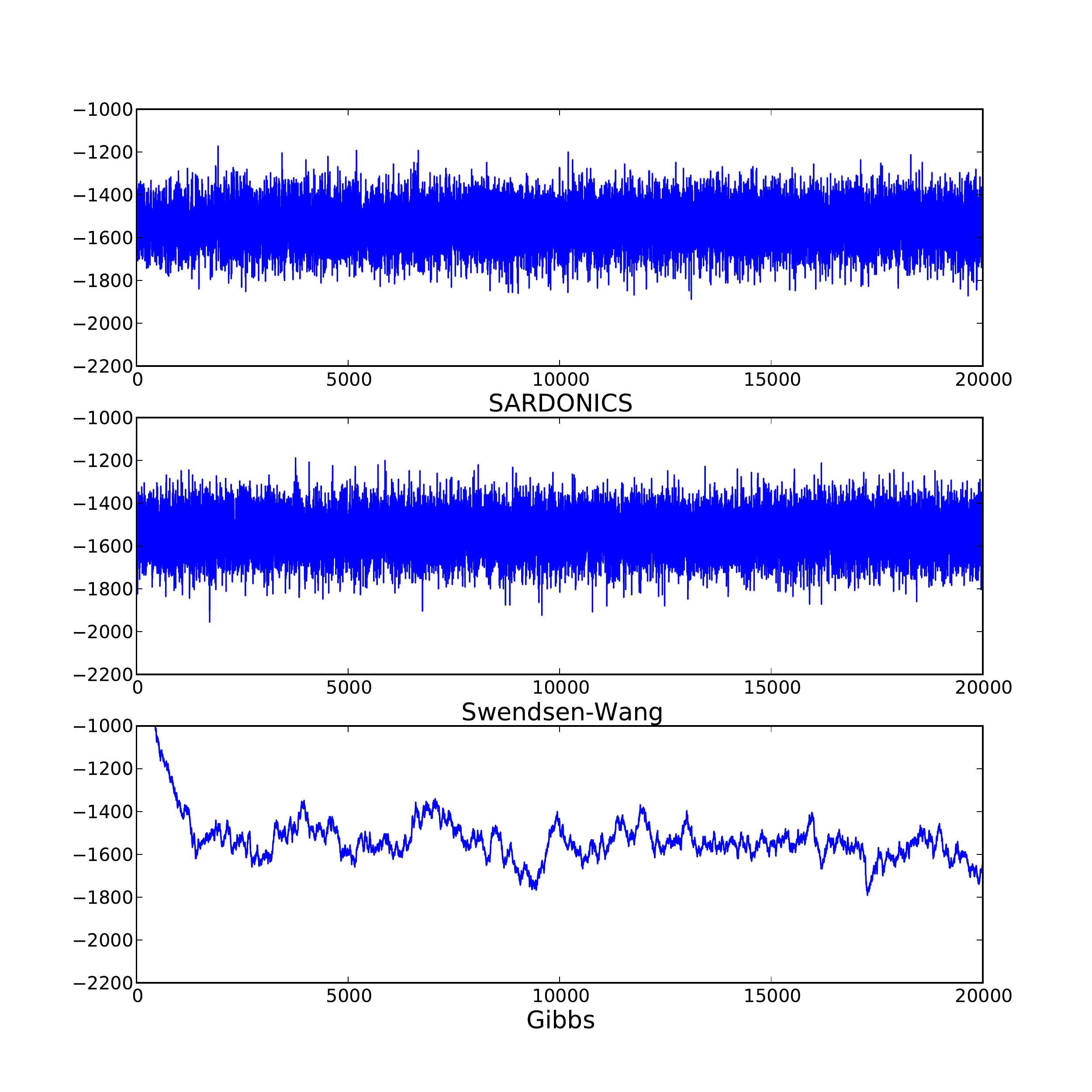} \\  
    \end{tabular*}
    \caption{Auto-correlations of the samplers, on the Ferromagnetic 2D grid Ising model, for $T=1$ [top left] and $T = 5$ [top right].
    Traces of every 5 out of the $10^5$ samples of the energy at $T=1$ [bottom left] and $T=5$ [bottom right].}
  \label{fig:trace_temp} 
\end{figure}

As shown in Figure~\ref{fig:acf1}, at the critical temperature, Swendsen-Wang does 
considerably better than its competitors. It is precisely for these types of lattice that Swendsen-Wang was designed
to mix efficiently, through effective clustering. This part of the result is therefore not surprising. However, we must consider the performance of SARDONICS carefully. Although it does much better than Gibbs, as expected, it under-performs in comparison to Swendsen-Wang. This seems to be a consequence of the fact that the probability distribution at this temperature has many peaks. SARDONICS, despite its large moves, can get trapped in these peaks for many iterations. At temperature 5, when the distribution is flattened,  
the performance of SARDONICS is comparable to that of Swendsen-Wang as depicted in Figure~\ref{fig:trace_temp}.
The same figure also shows the results for $T=1$, where the target distribution is even more peaked. The good performance of SARDONICS for $T=1$ might seem counterintuitive considering the previous results for $T=2.27$. 
An explanation is provided subsequently.

At temperatures $1$ and $2.27$, adaptation is very hard. Before the sampler converges, it is beneficial for it to take large steps to achieve lower auto-correlation. The adaptation mechanism learns this and hence automatically chooses large SAW lengths. But after the sampler converges, large steps can take the sampler out of the high-probability region thus leading to low acceptance. If the sampler hits the peak during adaptation then it learns to choose small SAW lengths. This may also be problematic if we restart the sampler from a different state. This points out one of the dangers of having finite adaptation schemes. A simple solution, in this particular case, is to change the bounds on the SAW lengths manually. This enables SARDONICS to achieve performance comparable to that of Swendsen Wang for 
these nearly deterministic models, as shown in Figure~\ref{fig:trace_temp} for $T=1$. However, ideally, infinite adaptation mechanisms might provide a more principled and general solution.

In addition to studying the effect of temperature changes on the performance of the algorithms, we also investigate their sensitivity with respect to the addition of unsatisfiable constraints. To accomplish this, we set the 
interaction weights $J_{ij}$ and the biases $h_i$ uniformly at random on the set $\{-1, 1\}$. We set the temperature to $T=1.0$. We refer to this model as the frustrated 2D grid Ising model. As shown by the auto-correlations and energy traces plotted in Figure~\ref{fig:2d}, SARDONICS does considerably better than its rivals. It is interesting to note that Swendsen-Wang does much worse on this model as the unsatisfiable constraints hinder effective clustering. The figure also shows the reward obtained by the Bayesian optimization scheme as a function of the number of adaptations. The adaptation algorithm traded-off exploration and exploitation effectively in this case.
\begin{figure}[t]
\begin{tabular*}{10 cm}{c c}
  \raisebox{0.1cm}{\includegraphics[trim = 35mm 15mm 35mm 35mm, clip, width=0.4\textwidth]{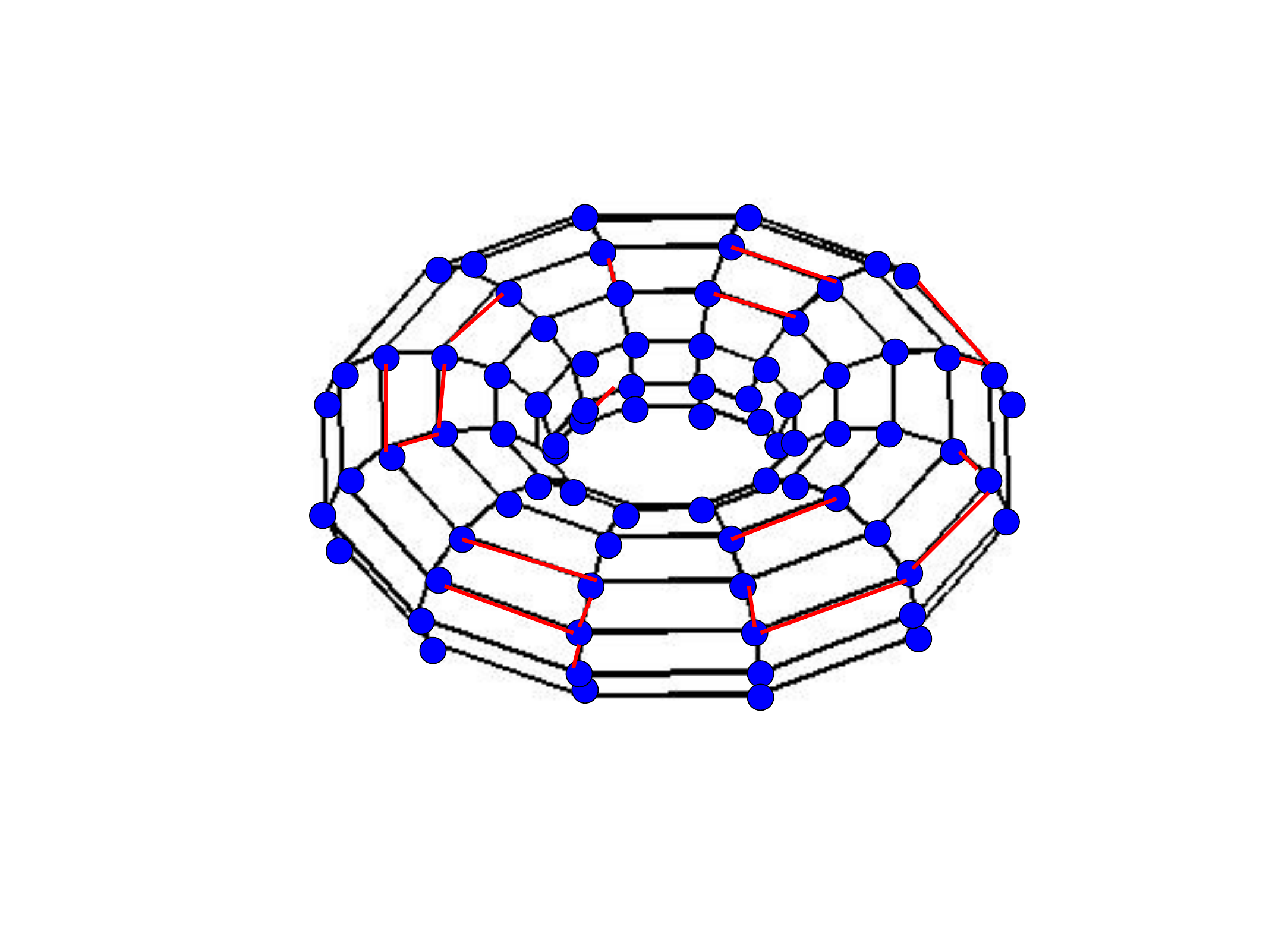}} &
  \includegraphics[width=0.55\textwidth]{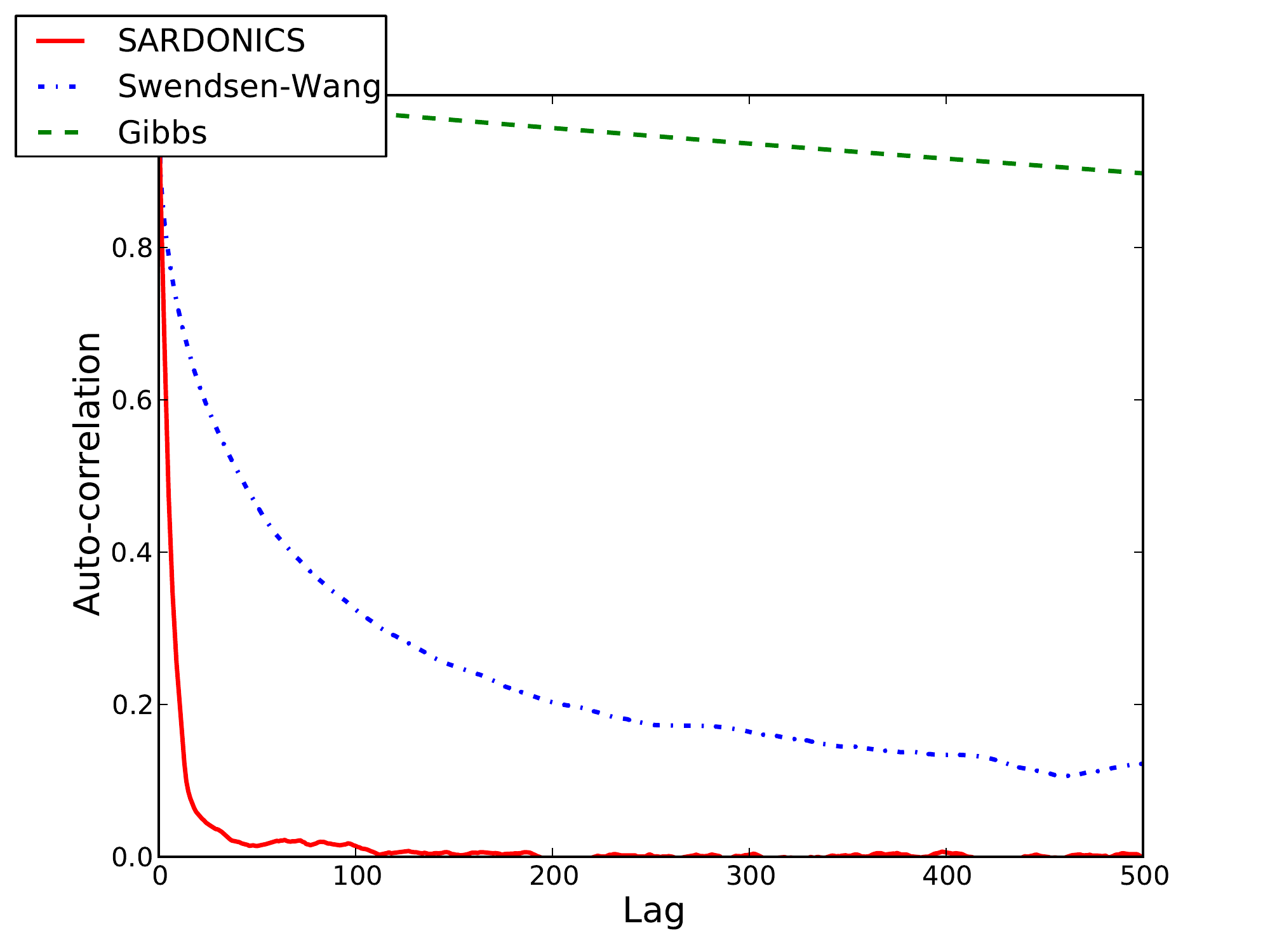} \\
    \includegraphics[width=0.45\textwidth]{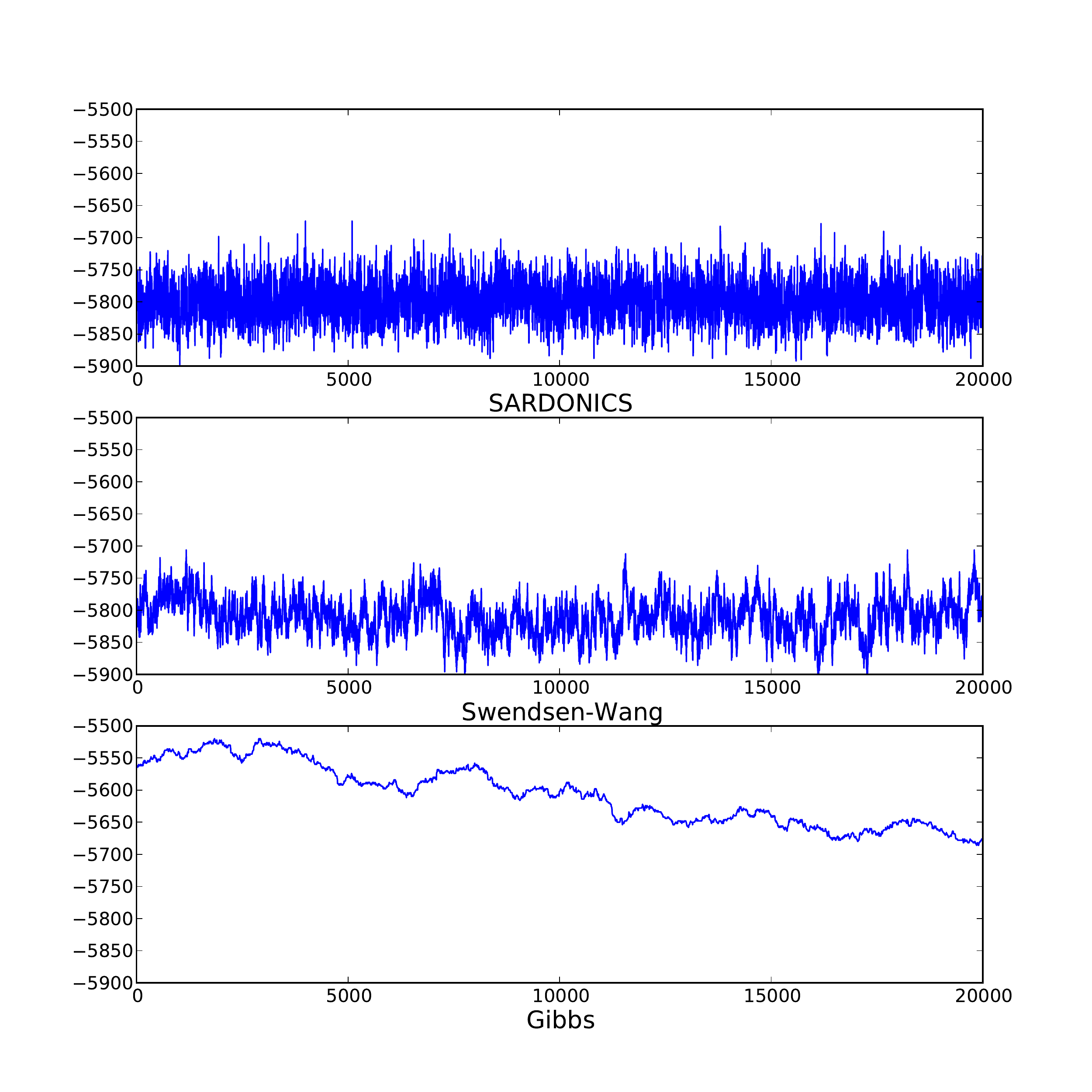} &
  \includegraphics[width=0.55\textwidth]{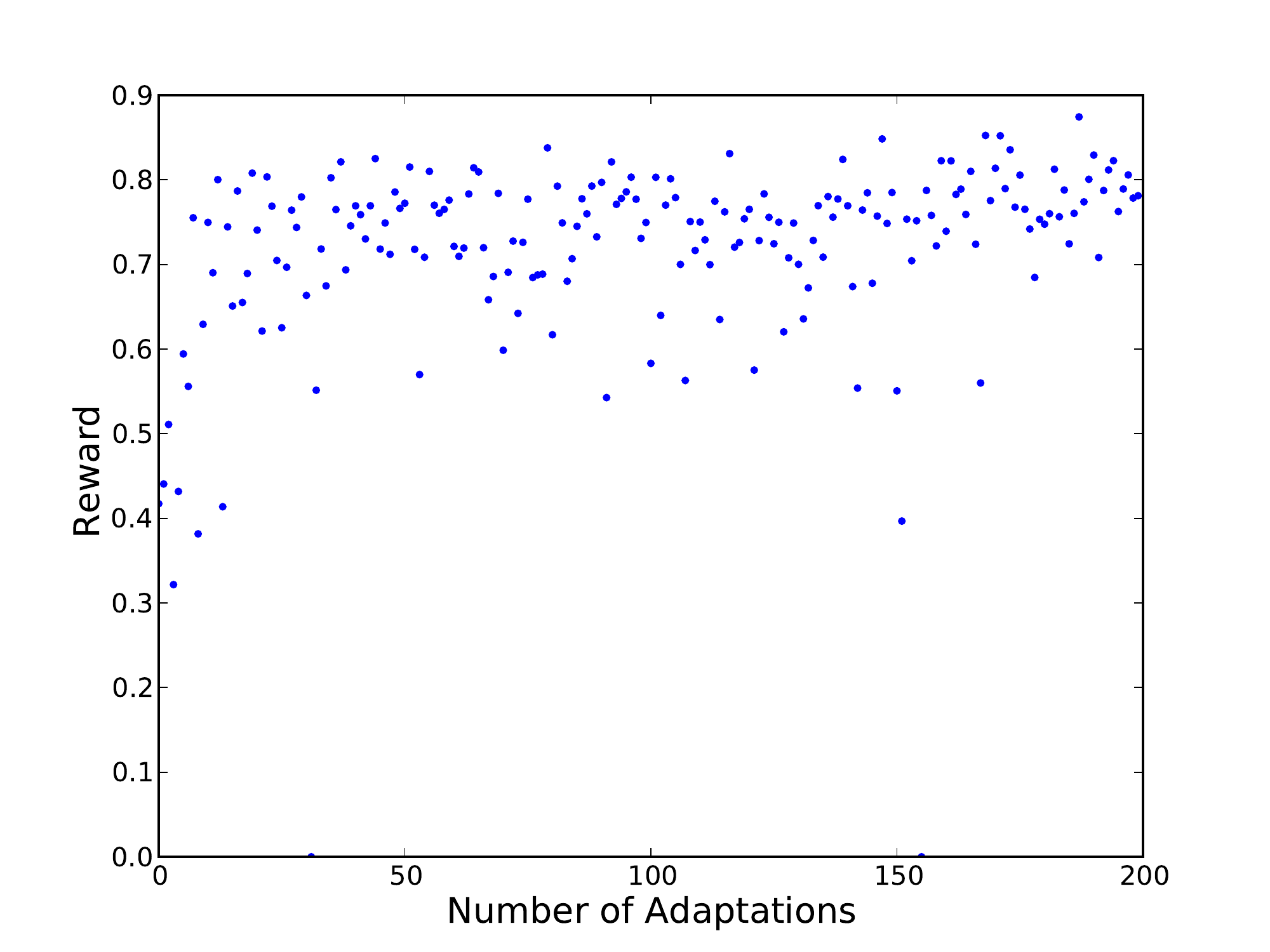}
\end{tabular*}
  \caption{Frustrated 2D grid Ising model with periodic boundaries [top left], auto-correlations of the three samplers [top right], traces of the last 20000 of 100000 samples of the energy [bottom left] and rewards obtained by the Bayesian optimization algorithm during the adaptation phase [bottom right].}
  \label{fig:2d}
\end{figure}

The third batch of experiments compares the algorithms on an
Ising model where the variables are topologically structured as a
$9\times 9 \times 9$ three-dimensional cube, \( J_{ij} \) are uniformly sampled from
the set \( \{-1,1\} \), and the \( h_{i} \) are zero. \( \beta \) was
set to \( 1.0 \), corresponding to a lower temperature than the value
of \( 0.9 \), at which it is known \cite{Marinari-97} that, roughly
speaking, regions of the state space become very difficult to visit
from one another via traditional Monte Carlo simulation. Figure~\ref{fig:3d} shows that for this more densely connected model, the performance of Swendsen-Wang deteriorates substantially. However, SARDONICS still mixes reasonably well.
\begin{figure}[t]
\begin{tabular*}{10 cm}{c c}
  \raisebox{0.5cm}{\includegraphics[width=0.3\textwidth]{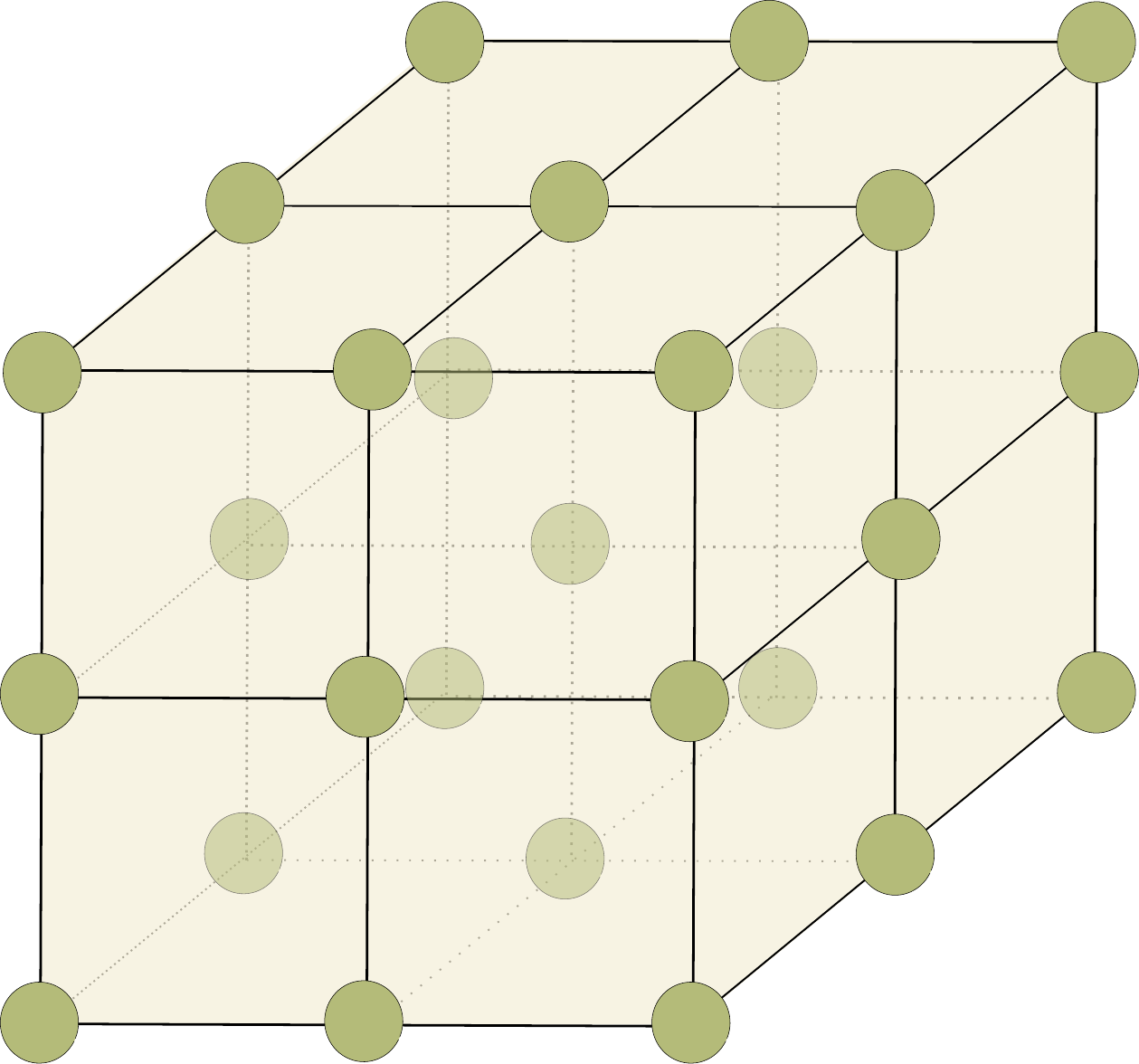}} &
  \includegraphics[width=0.55\textwidth]{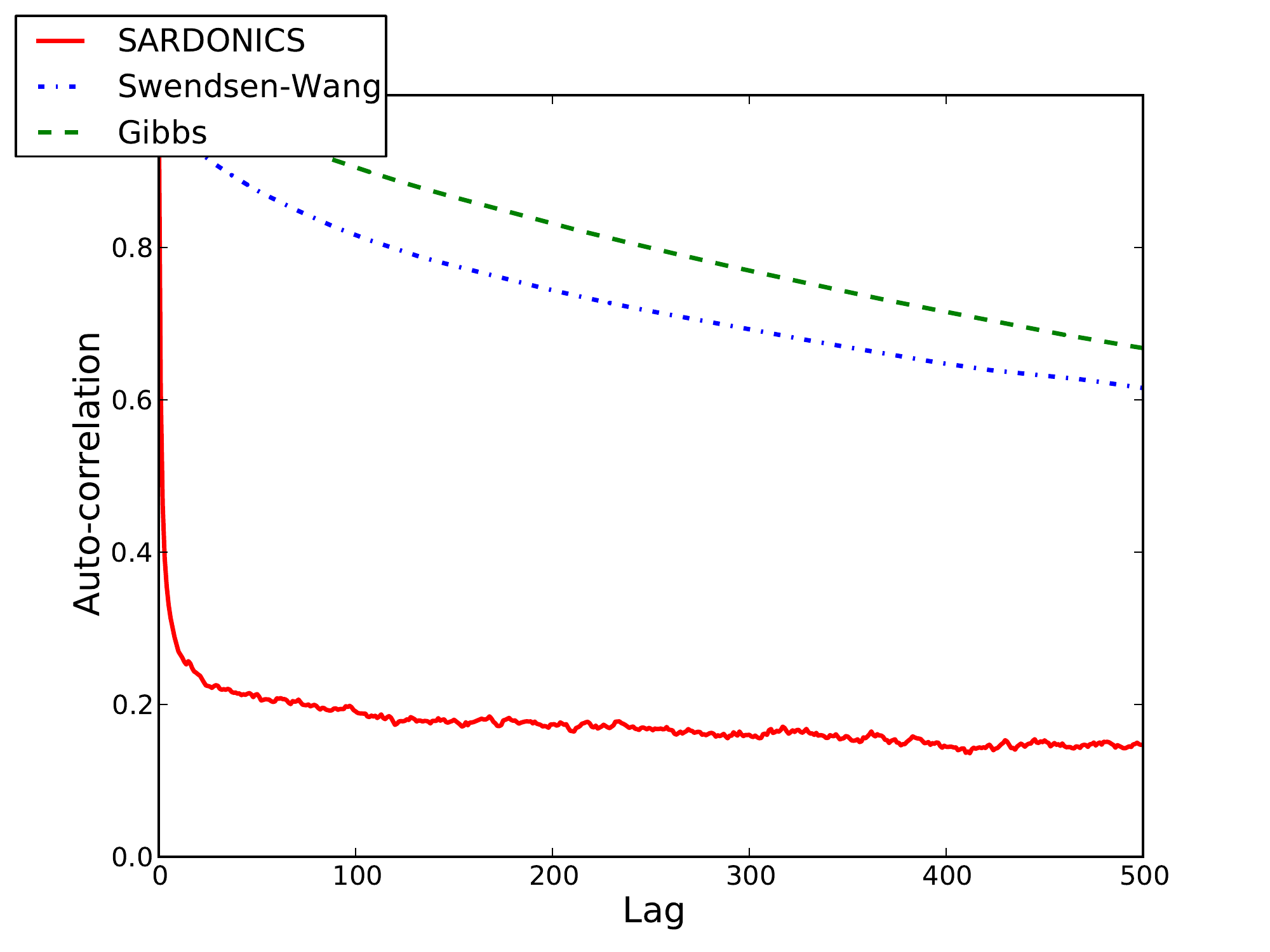} \\
    \includegraphics[width=0.45\textwidth]{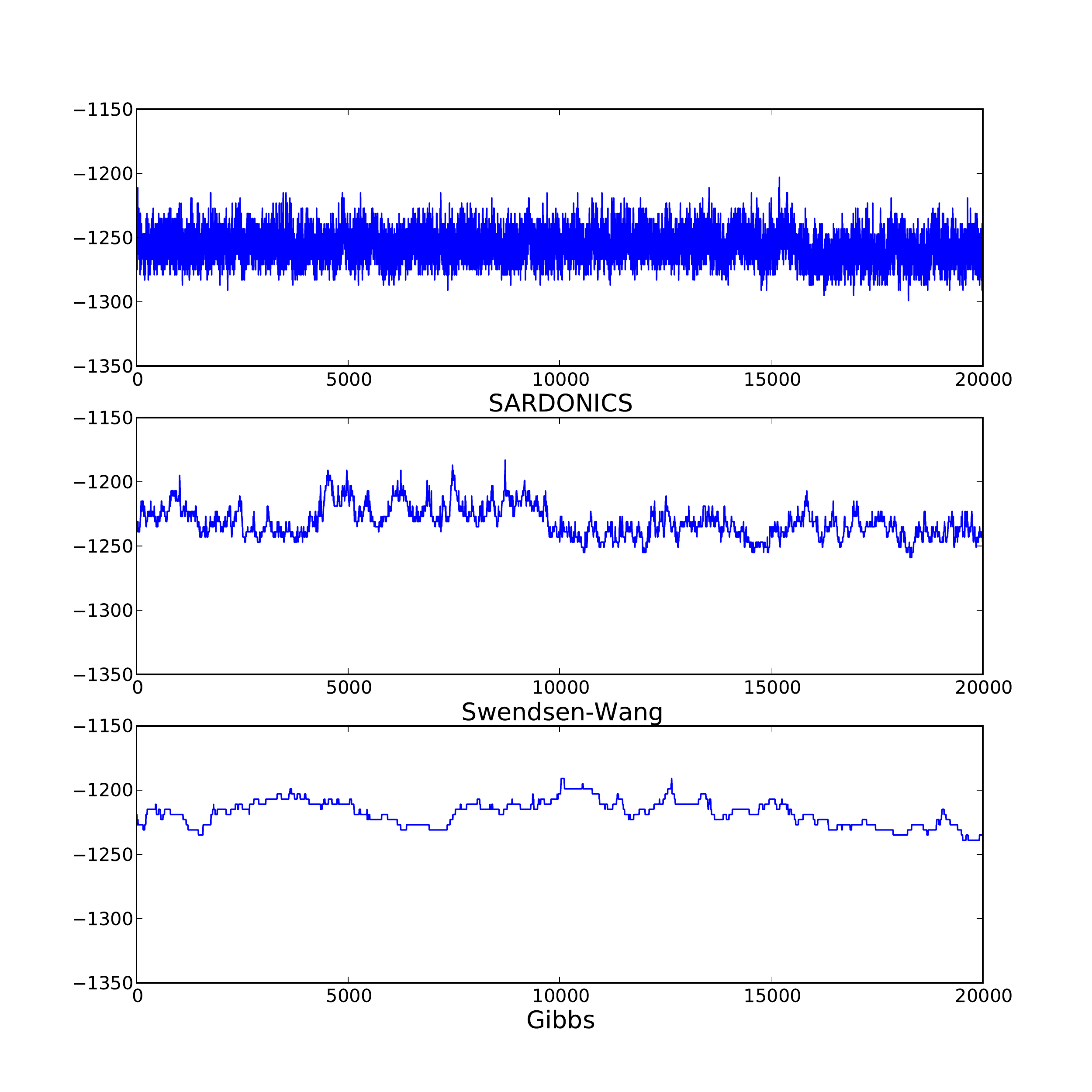} &
  \includegraphics[width=0.55\textwidth]{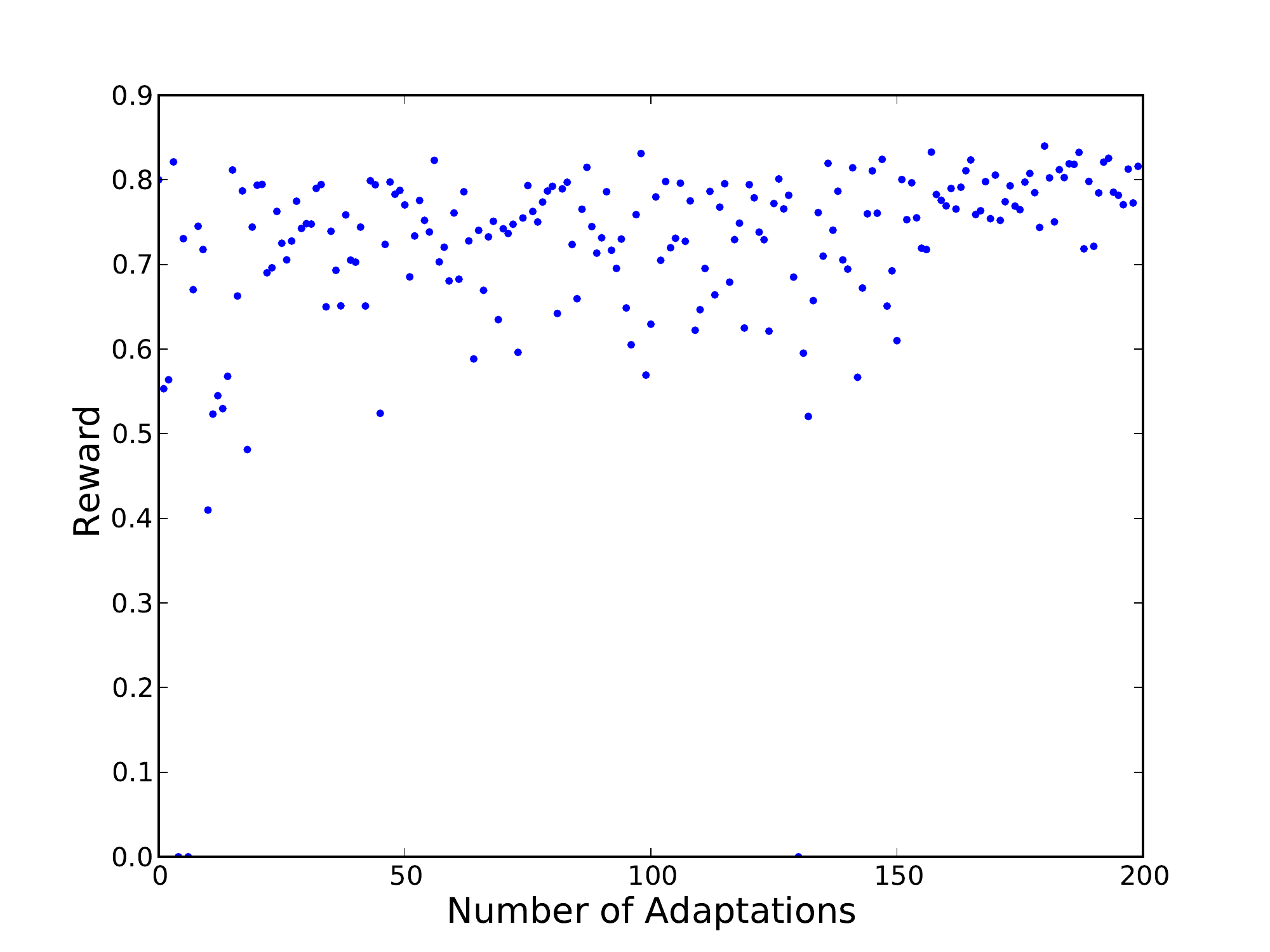}
\end{tabular*}
  \caption{Frustrated 3D cube Ising model with periodic boundaries (for visualization simplicity, the boundary edges are not shown) [top left], auto-correlations of the three samplers [top right], traces of the last 20000 of 100000 samples of the energy [bottom left] and rewards obtained by the Bayesian optimization algorithm during the adaptation phase [bottom right].}
  \label{fig:3d}
\end{figure}

While the three-dimensional-cube spin-glass is a much harder problem than the 2D ferromagnet, it represents a worst case scenario. One would hope that problems arising in practice will have structure in the potentials that would ease the problem of inference. For this reason, 
the third experimental set consisted of runs on a restricted
  Boltzmann machine \cite{Smolensky-86} with parameters trained from natural image patches via stochastic maximum likelihood \cite{Swersky-10,Marlin-10}.
RBMs are bipartite undirected probabilistic graphical models. The variables on one side are often referred to as ``visible units'', while the others are called ``hidden units''. Each visible unit is connected to all hidden units. However there are no connections among the hidden units and among the visible units. Therefore, given the visible units, the hidden units are conditionally independent and vice-versa.
Our model consisted of 784 visible and 500
hidden units. The model is illustrated in Figure~\ref{fig:rbm}.
\begin{figure}[h]
\begin{tabular*}{10 cm}{c c}
  \raisebox{1.5cm}{\includegraphics[width=0.4\textwidth]{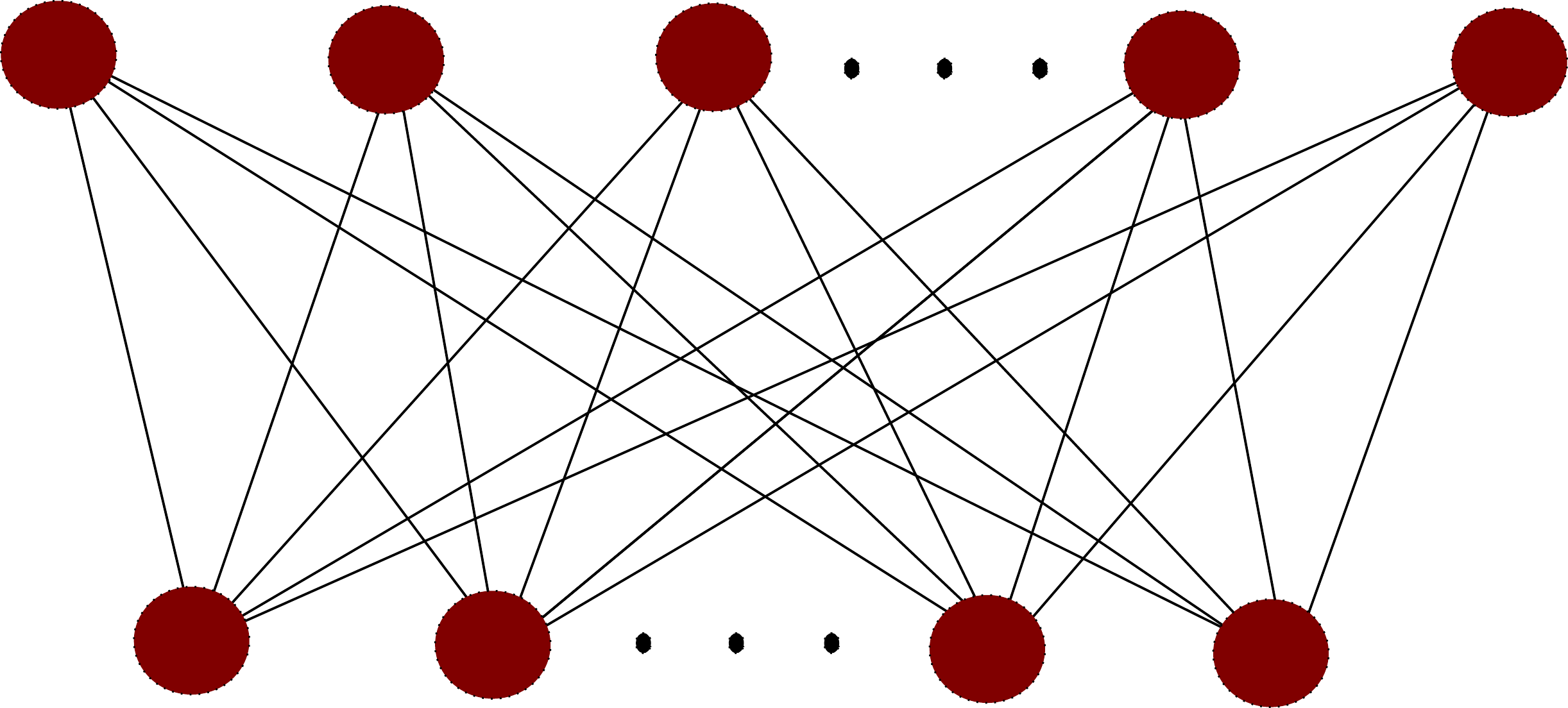}} &
  \includegraphics[width=0.55\textwidth]{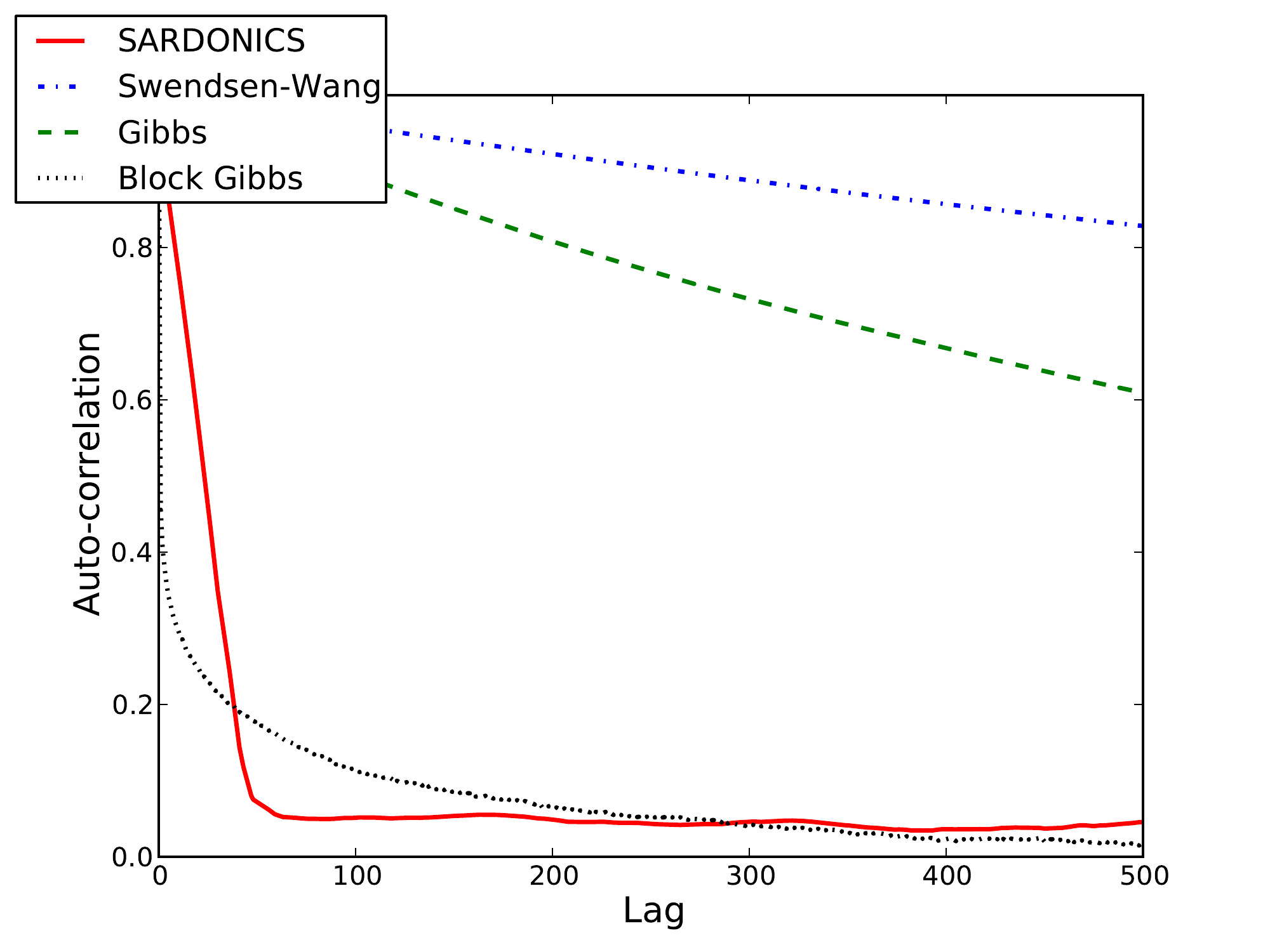} \\
    \includegraphics[width=0.45\textwidth]{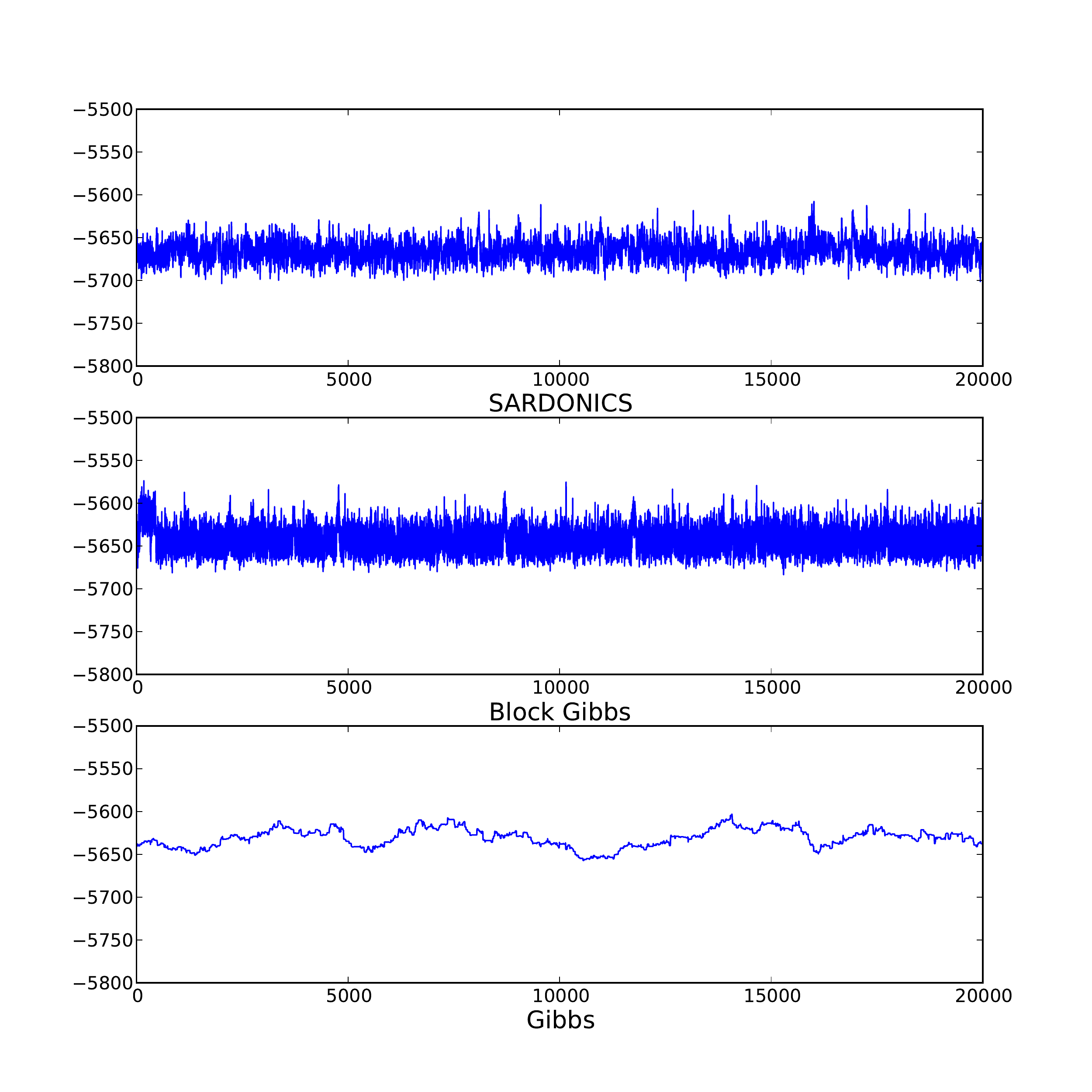} &
  \includegraphics[width=0.55\textwidth]{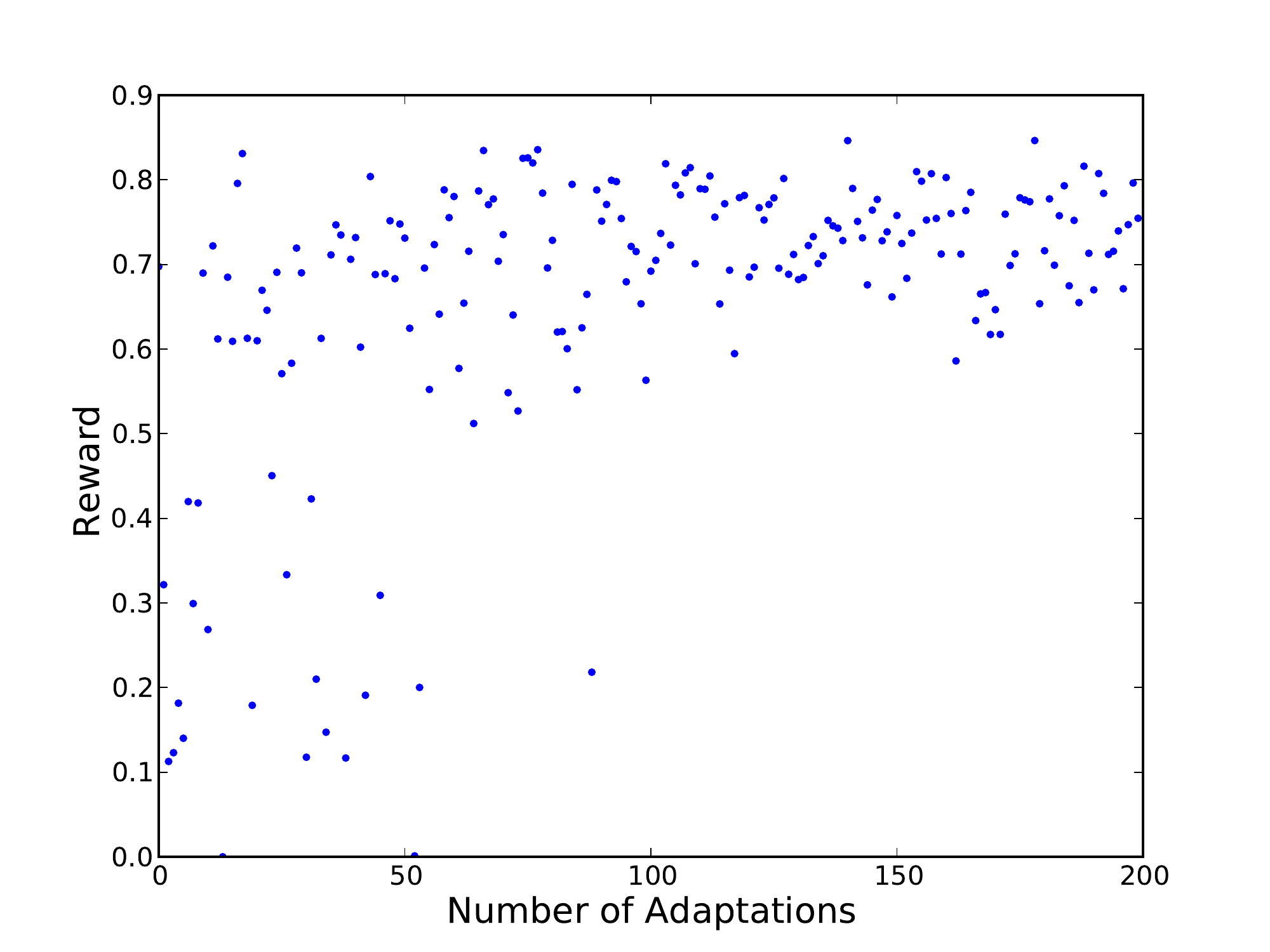}
\end{tabular*}
  \caption{Frustrated 3D cube Ising model with periodic boundaries [top left], auto-correlations of the three samplers [top right], traces of the last 20000 of 100000 samples of the energy [bottom left] and rewards obtained by the Bayesian optimization algorithm during the adaptation phase [bottom right].}
  \label{fig:rbm}
\end{figure}

\begin{figure}[h]
    \centering
  \includegraphics[width=0.5\textwidth]{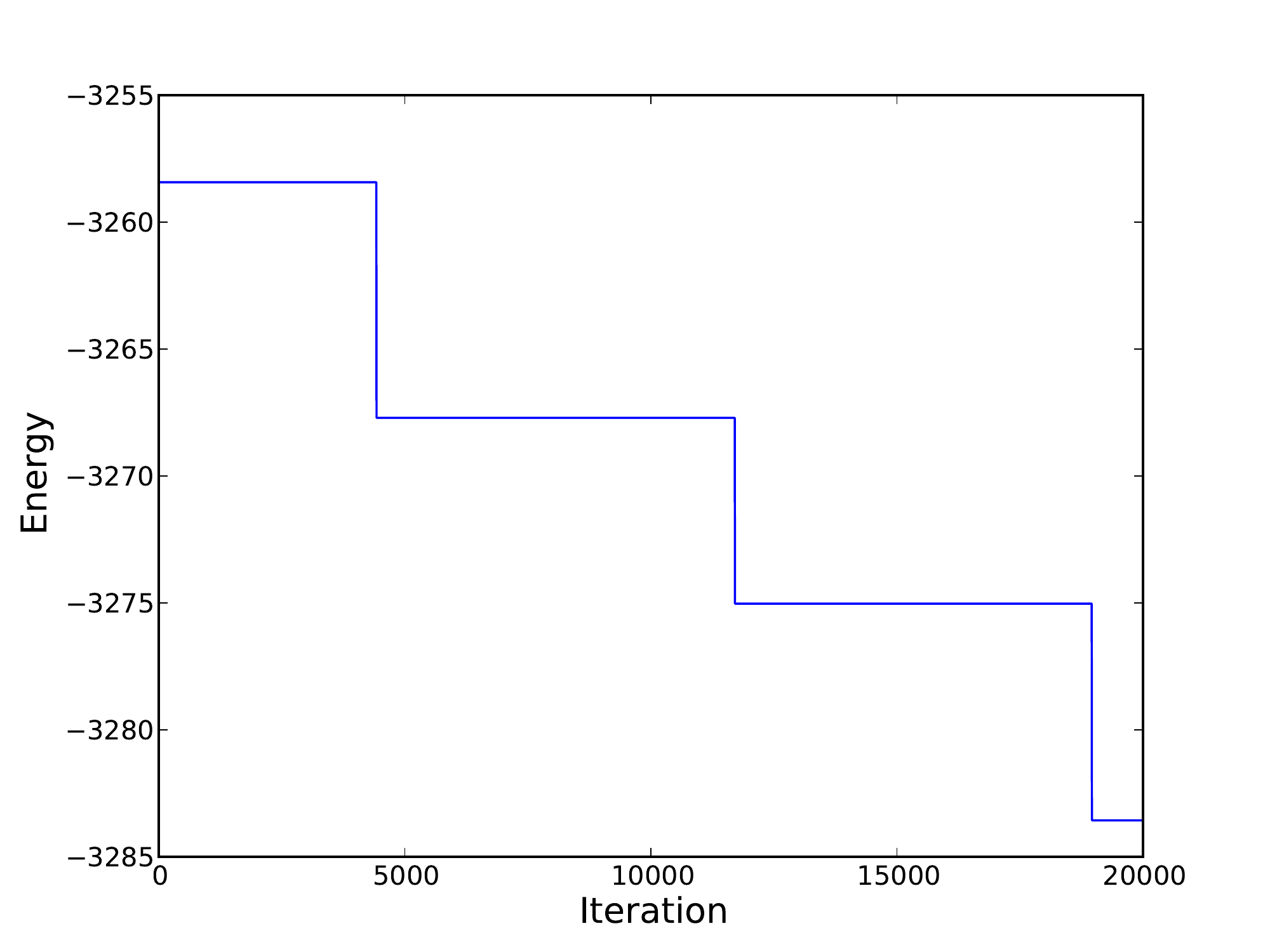} 
  \caption{Trace of the last 20000 of 100000 samples of the energy for the Swendsen-Wang sampler. As we can see, the sampler performs poorly on the RBM model.}
  \label{fig:rbmsw}
\end{figure}
\begin{figure}[h]
    \centering
  \includegraphics[width=0.45\textwidth]{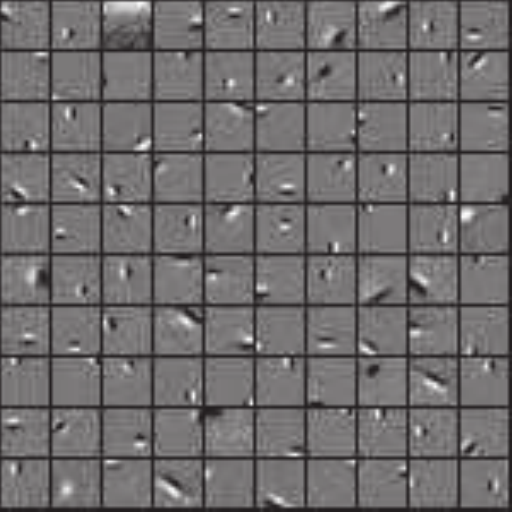} 
  \caption{RBM parameters. Each image corresponds to the parameters connecting a specific hidden unit to the entire set of visible units.}
  \label{fig:rbm2}
\end{figure}

The pre-learned interaction parameters capture local regularities of natural images \cite{Hyvarinen-09}. Some of these parameters are depicted as images in Figure~\ref{fig:rbm2}.
The parameter \( \beta \) was set to one. The total
number of variables and edges in the graph were thus 1284 and 392000 respectively. 

Figures~\ref{fig:rbm} and \ref{fig:rbmsw} show the results. Again SARDONICS mixes significantly better than Swendsen-Wang and the naive Gibbs sampler. Swendsen-Wang performs poorly on this model. As shown in Figure~\ref{fig:rbmsw}, it mixes slowly and fails to converge after $10^5$ iterations. 

For this bipartite model, it is possible to carry out block-Gibbs sampling (the standard method of choice). Encouragingly, SARDONICS compares well against this popular block strategy. This is important, because computational neuro-scientists would like to add lateral connections among the hidden units, in which case block-Gibbs would no longer apply unless the connections form a tree structure. SARDONICS thus promises to empower computational neuro-scientists to address more sophisticated models of perception. The catch is that at present SARDONICS takes considerably more time than block-Gibbs sampling for these models. We discuss this issue in greater length in the following section.

Finally, we consider a frustrated 128-bit chimera lattice that arises in the construction of quantum computers \cite{Bian-11}. As depicted in Figure~\ref{fig:chimera}, SARDONICS once again outperforms its competitors. 
\begin{figure}[h]
\begin{tabular*}{10 cm}{c c}
  \includegraphics[width=0.4\textwidth]{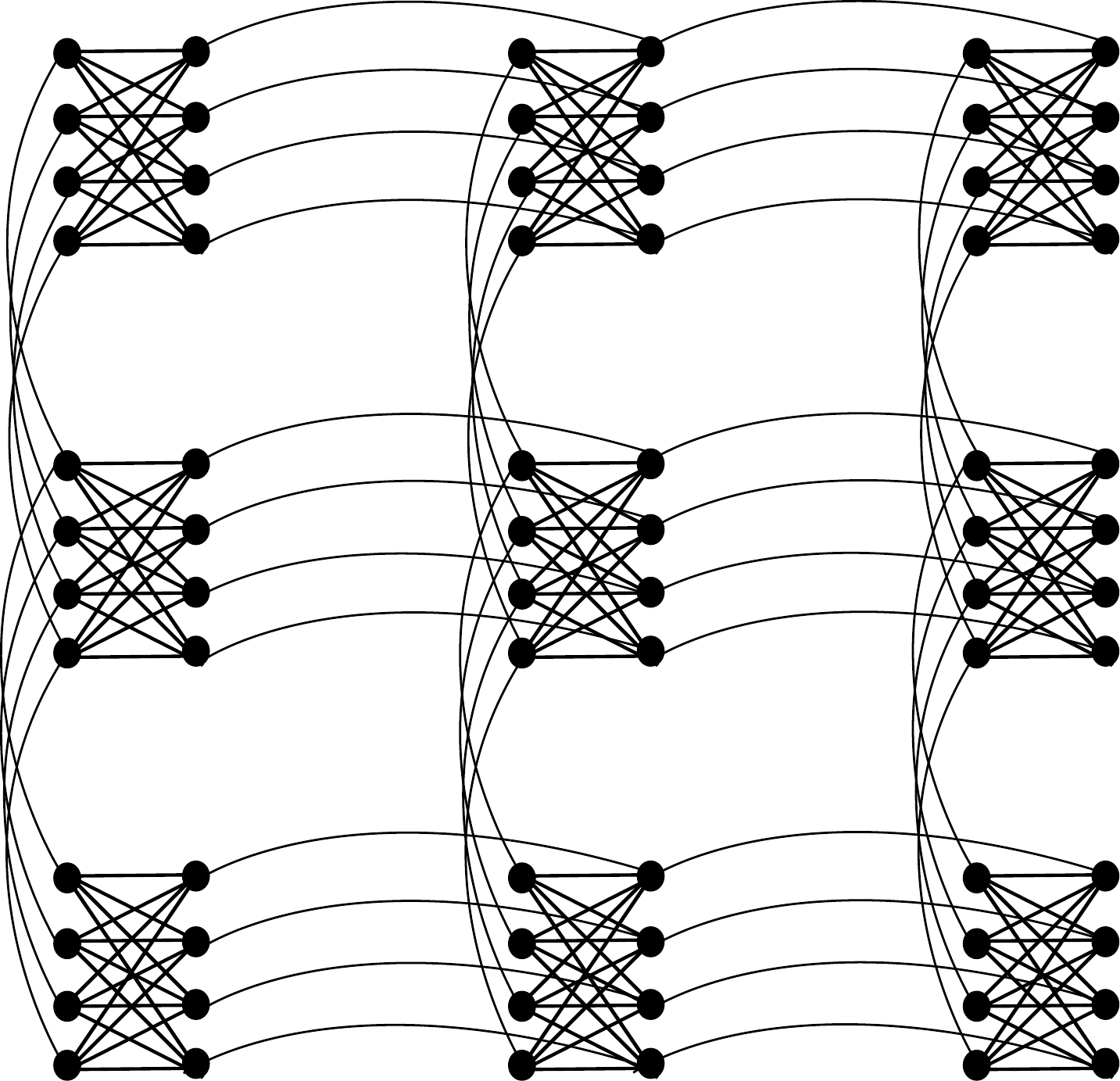} &
  \includegraphics[width=0.55\textwidth]{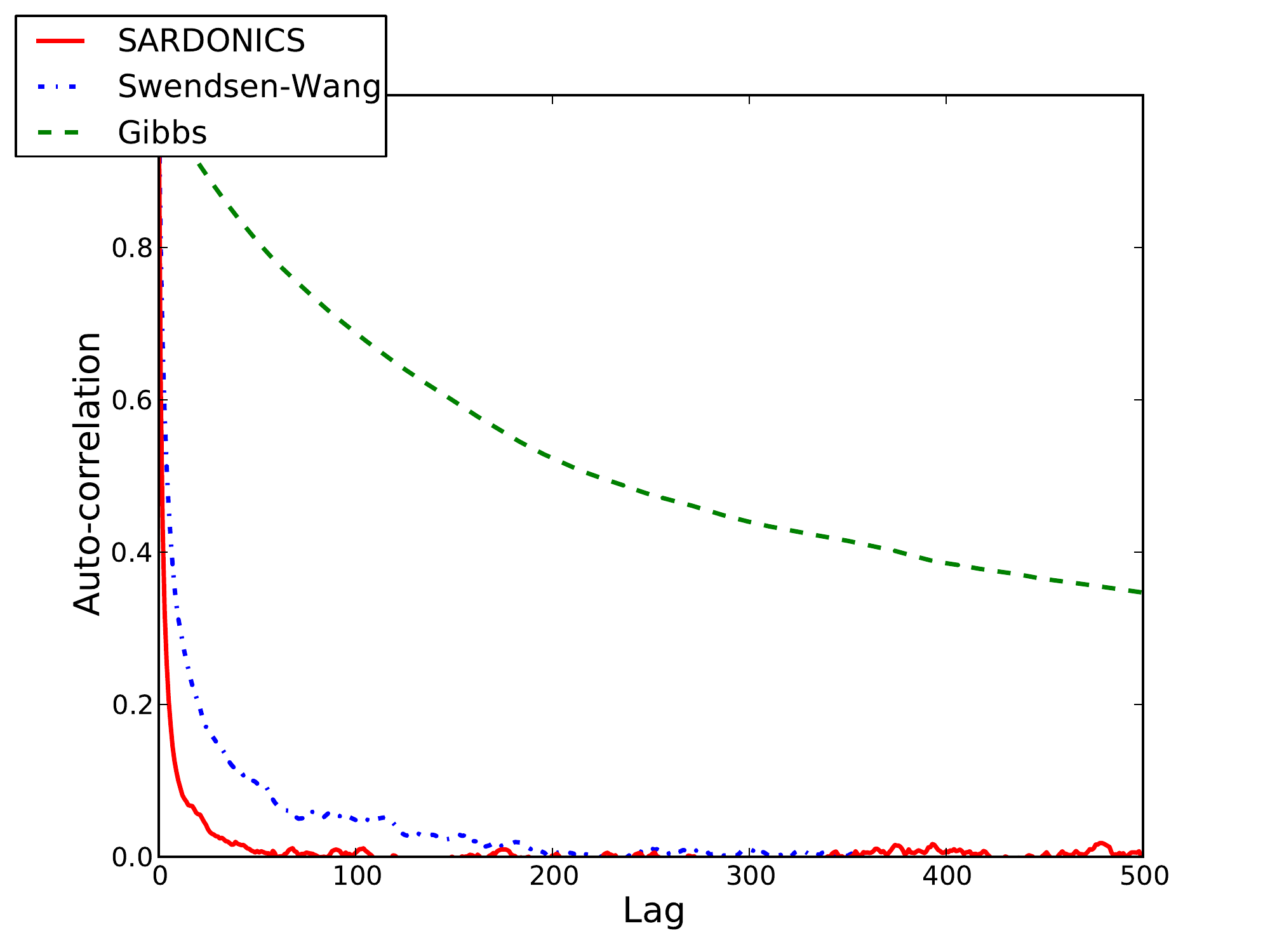} \\
    \includegraphics[width=0.45\textwidth]{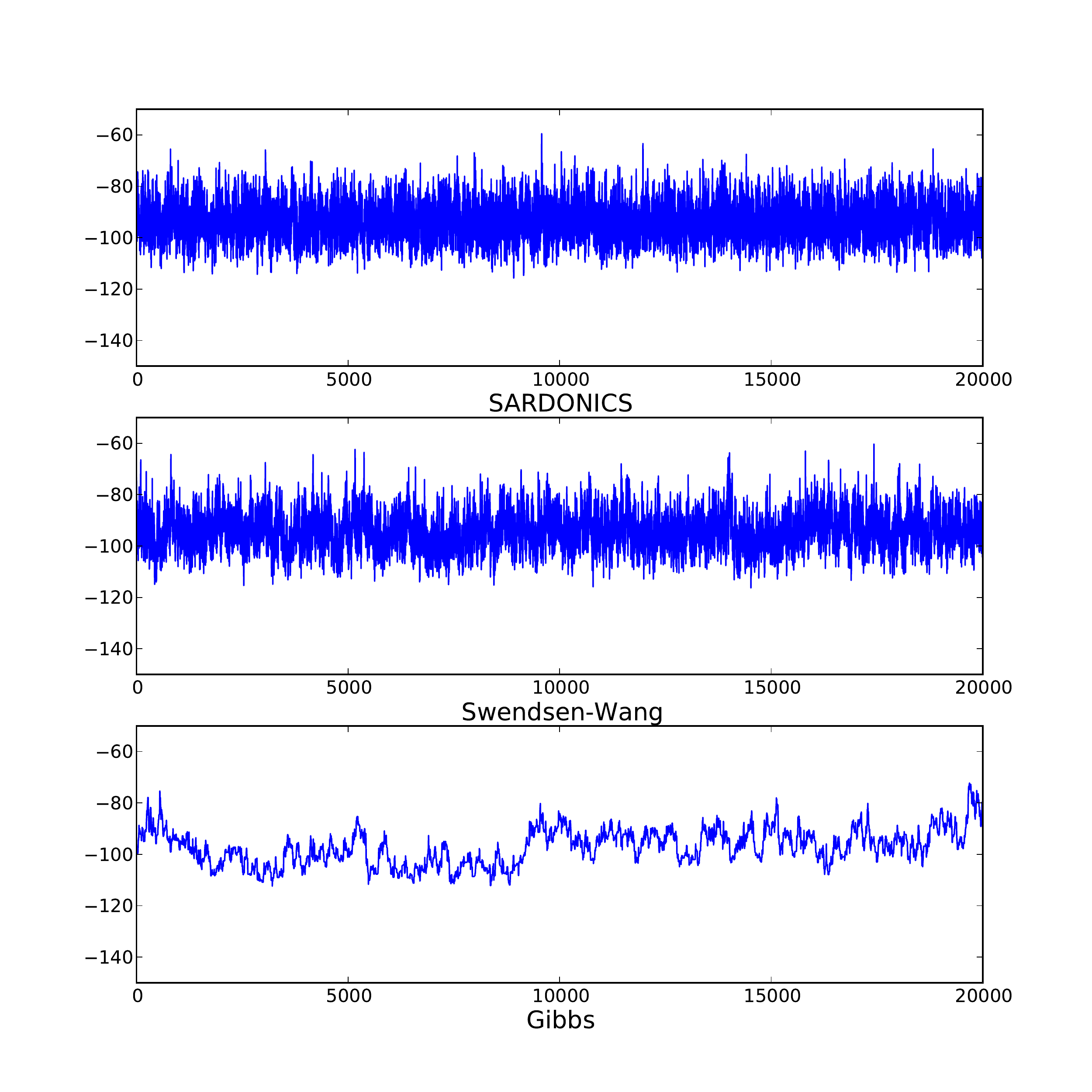} &
  \includegraphics[width=0.55\textwidth]{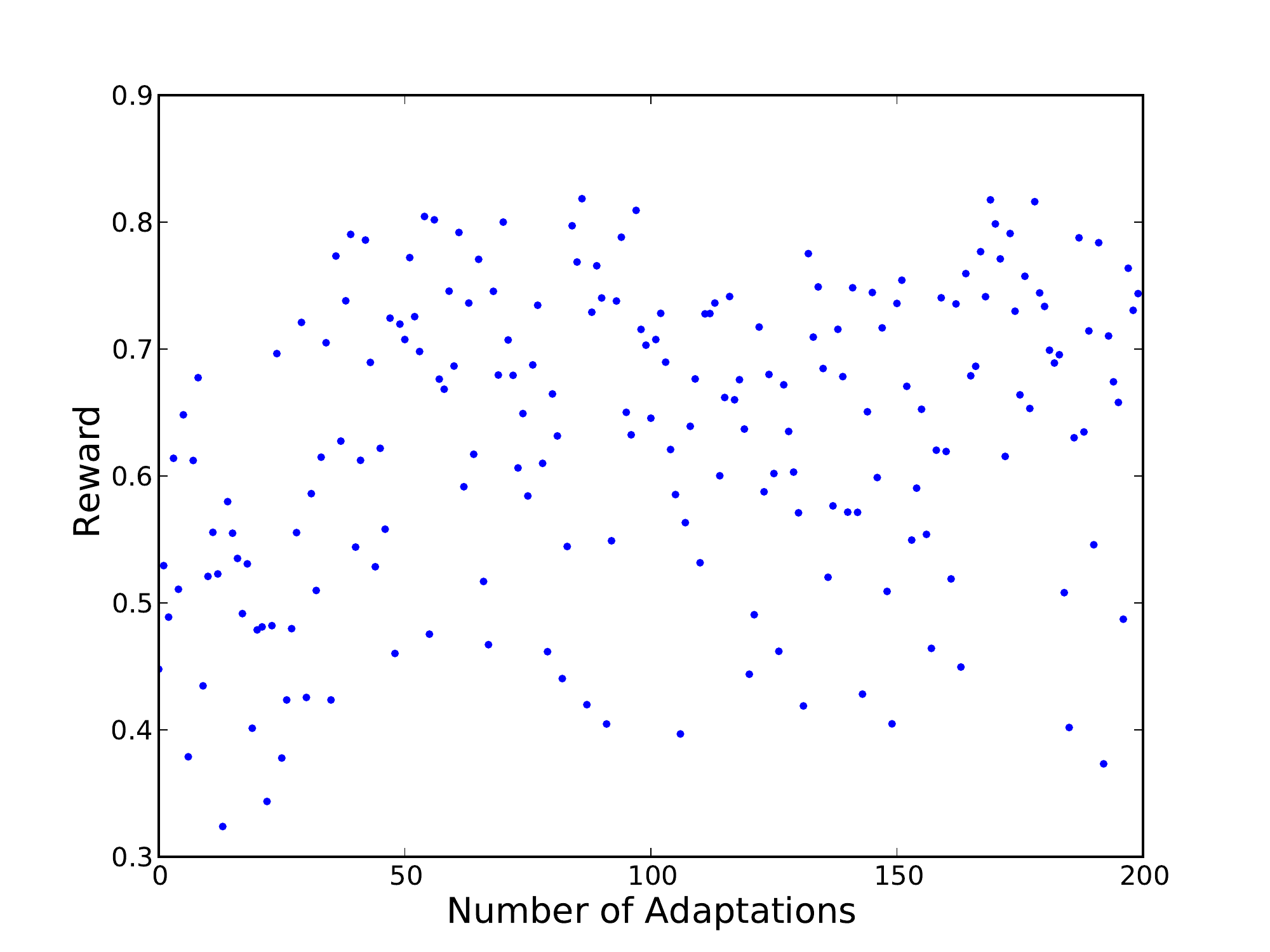}
\end{tabular*}
  \caption{Chimera lattice [top left], auto-correlations of the three samplers [top right], traces of the last 20000 of 100000 samples of the energy [bottom left] and rewards obtained by the Bayesian optimization algorithm during the adaptation phase [bottom right].}
  \label{fig:chimera}
\end{figure}

\subsection{Computational considerations}


\begin{table}[h!]

\label{table:sampler_time}
\begin{center}
\begin{tabular}{lrrrr}
& \multicolumn{4}{c}{\textbf{Samplers}} \\
\cline{2-5} \\
 \textbf{Models} &  SARDONICS & Swendsen-Wang & Gibbs & Block Gibbs \\
\hline \\
Ferromagnetic Ising Model & 20 minutes & 50 minutes & tens of seconds & N/A \\
Frustrated 2D Ising Model & 20 minutes & 50 minutes & tens of seconds & N/A \\
Frustrated 3D Ising Model & 10 minutes & a few minutes & tens of seconds & N/A \\
RBM & 10 hours & a few hours & 20 minutes & a few minutes \\
Chimera & a few minutes & tens of seconds & tens of seconds & N/A
\end{tabular}
\end{center}
\caption{Rough computation time for each sampler on different models. 
All samplers are run on the same computer with 8 CPUs. The Swendsen-Wang sampler is coded in Matlab with 
the computationally intensive part written in C. The SARDONICS sampler is coded in Python also with its
computationally intensive part coded in C. The Gibbs and block Gibbs samplers are coded in Python. 
The Swendsen-Wang and block Gibbs samplers take advantage of parallelism
via parallel numerical linear algebra operations. The SARDONICS sampler and the Gibbs sampler, however, run on a single CPU. The adaptation time of SARDONICS is also included.}
\end{table}

The bulk of the computational time of the SARDONICS algorithm is spent
in generating states with the SAW proposal. At each step of the
process, a component from a discrete probability vector, corresponding
to the variable to flip, must be sampled. Naively, the time needed to
do so scales linearly with the length \( l \) of the vector. In
graphical models of sparse connectivity, however, it is possible
achieve a dramatic computational speedup by storing the vector in a
\emph{binary heap}. Sampling from a heap is of \( O(\log l) \), but
for \emph{sparsely connected} models, updating the heap in response to
a flip, which entails replacing the energy changes that would result
if the flipped variable's \emph{neighboring} variables were themselves
to flip, is also of logarithmic complexity. In contrast, for a densely
connected model, the heap update would be of \( O(M\log l) \) where $M$
is the maximum degree of the graph, while
recomputing the probability vector in the naive method is \( O(M)\). 
The simple method is thus cheaper for dense models.
In our experiments, we implemented SARDONICS with the binary heap despite the fact that the
RBM is a densely connected model.

For densely connected models, one could easily parallelize the computation of 
generating states with the SAW proposal. Suppose we have $n$ parallel processes.
Each process $P$ holds one section $U_p$ of the unnormalized probability vector $U$.
To sample from the discrete vector in parallel, each parallel process can
sample one variable $v_p$ to flip according to $U_p$. Each process also calculates 
the sum of its section of the unnormalized probability vector $s_p$.
Then the variable to flip is sampled from the set $\{v_p: 1 \leq p \leq n\}$
proportional to the discrete probability vector $[s_1, s_2,..., s_n]$. To update 
$U$ in response to a flip, we could use the naive method mentioned above.
If we divide the $U$ evenly among processes and assume
equal processor speed, sampling from and updating the unnormalized probability
vector would take \( O (\frac{M+l}{n} + n) \) operations. When $n$ is small compared
to $M + l$, we could achieve near linear speed up.

We compare, in table II, the amount of time it takes to draw $10^5$ samples using different 
samplers. In principle, it is unwise to compare the computational time of samplers when this comparison depends on specific implementation details and architecture. The table simply provides rough estimates.
The astute reader might be thinking that we could run Gibbs for the same time as SARDONICS and expect similar performance. This is however not what happens. We tested this hypothesis and found Gibbs to still underperform. Moreover, in some cases, Gibbs can get trapped as a result of being a small move sampler.

The computational complexity of SARDONICS is affected by the degree of the graph, whereas the complexity of Swendsen-Wang increases with the number of edges in the graph. So for large sparsely-connected graphs, we expect SARDONICS to have a competing edge over Swendsen-Wang, as already demonstrated to some extent in the experiments.

%% file: SARDONICS.bbl
\begin{thebibliography}{10}

\bibitem{Ackley-85}
D.~H. Ackley, G.~Hinton, and T.. Sejnowski.
\newblock A learning algorithm for {Boltzmann} machines.
\newblock {\em Cognitive Science}, 9:147--169, 1985.

\bibitem{Andrieu-01}
Christophe Andrieu and Christian Robert.
\newblock {Controlled MCMC for optimal sampling}.
\newblock Technical Report 0125, Cahiers de Mathematiques du Ceremade,
  Universite Paris-Dauphine, 2001.

\bibitem{Barahona-82}
F~Barahona.
\newblock On the computational complexity of {Ising} spin glass models.
\newblock {\em Journal of Physics A: Mathematical and General}, 15(10):3241,
  1982.

\bibitem{Berg-91}
Bernd~A. Berg and Thomas Neuhaus.
\newblock Multicanonical algorithms for first order phase transitions.
\newblock {\em Physics Letters B}, 267(2):249 -- 253, 1991.

\bibitem{Bertsekas-87}
Dimitri~P. Bertsekas.
\newblock Projected {Newton} methods for optimization problems with simple
  constraints.
\newblock {\em SIAM Journal on Control and Optimization}, 20(2):221--246, 1982.

\bibitem{Besag-74}
Julian Besag.
\newblock Spatial interaction and the statistical analysis of lattice systems.
\newblock {\em J. Roy. Statist. Soc., Ser. B}, 36:192--236, 1974.

\bibitem{Bian-11}
Zhengbing Bian, Fabian Chudak, William~G. Macready, and Geordie Rose.
\newblock The {Ising} model: teaching an old problem new tricks.
\newblock Technical report, D-Wave Systems, 2011.

\bibitem{Blei-03}
David~M. Blei, Andrew~Y. Ng, and Michael~I. Jordan.
\newblock Latent {Dirichlet} allocation.
\newblock {\em J. Mach. Learn. Res.}, 3:993--1022, 2003.

\bibitem{Brochu-09}
Eric Brochu, Vlad~M Cora, and Nando {de Freitas}.
\newblock A tutorial on {B}ayesian optimization of expensive cost functions.
\newblock Technical Report TR-2009-023, University of British Columbia,
  Department of Computer Science, 2009.

\bibitem{Bull-11}
Adam~D Bull.
\newblock Convergence rates of efficient global optimization algorithms.
\newblock Technical Report arXiv:1101.3501v2, 2011.

\bibitem{Diggle-98}
P.~J. Diggle, J.~A. Tawn, and R.~A. Moyeed.
\newblock Model-based geostatistics.
\newblock {\em Journal of the Royal Statistical Society. Series C (Applied
  Statistics)}, 47(3):299--350, 1998.

\bibitem{Duane-87}
S~Duane, A~D Kennedy, B~J Pendleton, and D~Roweth.
\newblock Hybrid {Monte Carlo}.
\newblock {\em Physics Letters B}, 195(2):216--222, 1987.

\bibitem{Earl-05}
David~J. Earl and Michael~W. Deem.
\newblock Parallel tempering: Theory{,} applications{,} and new perspectives.
\newblock {\em Phys. Chem. Chem. Phys.}, 7:3910--3916, 2005.

\bibitem{Finkel-03}
Daniel~E Finkel.
\newblock {\em {DIRECT} Optimization Algorithm User Guide}.
\newblock Center for Research in Scientific Computation, North Carolina State
  University, 2003.

\bibitem{deFreitas-05}
Nando~De Freitas, Yang Wang, Maryam Mahdaviani, and Dustin Lang.
\newblock Fast {Krylov} methods for {N}-body learning.
\newblock In Y.~Weiss, B.~Sch\"{o}lkopf, and J.~Platt, editors, {\em Advances
  in Neural Information Processing Systems}, pages 251--258. MIT Press, 2005.

\bibitem{Geyer-91}
Charles~J Geyer.
\newblock {Markov Chain Monte Carlo} maximum likelihood.
\newblock In {\em Computer Science and Statistics: 23rd Symposium on the
  Interface}, pages 156--163, 1991.

\bibitem{Girolami-11}
Mark Girolami and Ben Calderhead.
\newblock Riemann manifold {Langevin} and {Hamiltonian Monte Carlo} methods.
\newblock {\em Journal of the Royal Statistical Society: Series B (Statistical
  Methodology)}, 73(2):123--214, 2011.

\bibitem{Glover-89}
F.~Glover.
\newblock Tabu search -- {P}art {I}.
\newblock {\em ORSA Journal on Computing}, 1:190--206, 1989.

\bibitem{Gramacy-04}
Robert~B Gramacy, Herbert K.~H. Lee, and William MacReady.
\newblock Parameter space exploration with gaussian process trees.
\newblock In {\em Proceedings of the International Conference on Machine
  Learning}, pages 353--360. Omnipress and ACM Digital Library, 2004.

\bibitem{Gubernatis-00}
J.~Gubernatis and N.~Hatano.
\newblock The multicanonical monte carlo method.
\newblock {\em Computing in Science Engineering}, 2(2):95--102, 2000.

\bibitem{Haario-01}
Heikki Haario, Eero Saksman, and Johanna Tamminen.
\newblock An adaptive {Metropolis} algorithm.
\newblock {\em Bernoulli}, 7(2):223--242, 2001.

\bibitem{Halko-11}
N.~Halko, P.~G. Martinsson, , and J.~A. Tropp.
\newblock Finding structure with randomness: Probabilistic algorithms for
  constructing approximate matrix decompositions.
\newblock {\em Science}, 53(2):217--288, 2011.

\bibitem{Hamze-04}
Firas Hamze and Nando de~Freitas.
\newblock From fields to trees.
\newblock In {\em Uncertainty in Artificial Intelligence}, pages 243--250,
  2004.

\bibitem{Hamze-05}
Firas Hamze and Nando de~Freitas.
\newblock {H}ot {C}oupling: a particle approach to inference and normalization
  on pairwise undirected graphs.
\newblock {\em Advances in Neural Information Processing Systems}, 18:491--498,
  2005.

\bibitem{Hamze-07}
Firas Hamze and Nando de~Freitas.
\newblock Large-flip importance sampling.
\newblock In {\em Uncertainty in Artificial Intelligence}, pages 167--174,
  2007.

\bibitem{Hamze-10}
Firas Hamze and Nando de~Freitas.
\newblock Intracluster moves for constrained discrete-space {MCMC}.
\newblock In {\em Uncertainty in Artificial Intelligence}, pages 236--243,
  2010.

\bibitem{Hinton-06a}
Geoffrey Hinton and Ruslan Salakhutdinov.
\newblock Reducing the dimensionality of data with neural networks.
\newblock {\em Science}, 313(5786):504--507, 2006.

\bibitem{Hoffman-11}
Matthew Hoffman, Eric Brochu, and Nando de~Freitas.
\newblock Portfolio allocation for {Bayesian} optimization.
\newblock In {\em Uncertainty in Artificial Intelligence}, pages 327--336,
  2011.

\bibitem{Hoos-04}
Holger~H. Hoos and Thomas Stutzle.
\newblock {\em Stochastic Local Search: Foundations and Applications}.
\newblock Elsevier, Morgan Kaufmann, 2004.

\bibitem{Hopfield-84}
J~J Hopfield.
\newblock Neurons with graded response have collective computational properties
  like those of two-state neurons.
\newblock {\em Proceedings of the National Academy of Sciences},
  81(10):3088--3092, 1984.

\bibitem{Hyvarinen-09}
A.~Hyvarinen, J.~Hurri, and P.O. Hoyer.
\newblock {\em Natural Image Statistics}.

\bibitem{Kindermann-80}
Ross Kindermann and J.~Laurie Snell.
\newblock {\em Markov Random Fields and their Applications}.
\newblock Amer. Math. Soc., 1980.

\bibitem{Kueck-05}
Hendrik {K\"{u}ck} and Nando de~Freitas.
\newblock Learning about individuals from group statistics.
\newblock In {\em Uncertainty in Artificial Intelligence}, pages 332--339,
  2005.

\bibitem{Lee-09}
Honglak Lee, Roger Grosse, Rajesh Ranganath, and Andrew~Y. Ng.
\newblock Convolutional deep belief networks for scalable unsupervised learning
  of hierarchical representations.
\newblock In {\em International Conference on Machine Learning}, pages
  609--616, 2009.

\bibitem{Liu-01}
Jun~S. Liu.
\newblock {\em {Monte Carlo} strategies in scientific computing}.
\newblock Springer, 2001.

\bibitem{Liu-03}
Jun~S. Liu, Junni~L. Zhang, Michael~J. Palumbo, and Charles~E. Lawrence.
\newblock Bayesian clustering with variable and transformation selections.
\newblock {\em Bayesian Statistics}, 7:249--275, 2003.

\bibitem{Lizotte-08}
Daniel Lizotte.
\newblock {\em Practical {Bayesian} Optimization}.
\newblock PhD thesis, University of Alberta, Edmonton, Alberta, Canada, 2008.

\bibitem{Mahendran-11}
Nimalan Mahendran, Ziyu Wang, Firas Hamze, and Nando de~Freitas.
\newblock {Bayesian} optimization for adaptive {MCMC}.
\newblock Technical Report arXiv:1110.6497v1, 2011.

\bibitem{Marinari-97}
E.~Marinari, G.~Parisi, and JJ~Ruiz-Lorenzo.
\newblock Numerical simulations of spin glass systems.
\newblock {\em Spin Glasses and Random Fields}, pages 59--98, 1997.

\bibitem{Marlin-10}
Benjamin Marlin, Kevin Swersky, Bo~Chen, and Nando {de Freitas}.
\newblock Inductive principles for restricted {Boltzmann} machine learning.
\newblock In {\em Artificial Intelligence and Statistics}, pages 509--516,
  2010.

\bibitem{May-11}
Benedict~C May, Nathan Korda, Anthony Lee, and David~S Leslie.
\newblock Optimistic {Bayesian} sampling in contextual bandit problems.
\newblock 2011.

\bibitem{Memisevic-09}
Roland Memisevic and Geoffrey Hinton.
\newblock Learning to represent spatial transformations with factored
  higher-order {Boltzmann} machines.
\newblock {\em Neural Computation}, 22:1473--1492, 2009.

\bibitem{Mockus-82}
Jonas Mo{\v c}kus.
\newblock The {B}ayesian approach to global optimization.
\newblock In {\em System Modeling and Optimization}, volume~38, pages 473--481.
  Springer Berlin / Heidelberg, 1982.

\bibitem{Munoz-03}
J.~D. Munoz, M.~A. Novotny, and S.~J. Mitchell.
\newblock Rejection-free {Monte Carlo} algorithms for models with continuous
  degrees of freedom.
\newblock {\em Phys. Rev. E}, 67, 2003.

\bibitem{Neal-10}
Radford~M Neal.
\newblock {MCMC} using {Hamiltonian} dynamics.
\newblock {\em Handbook of Markov Chain Monte Carlo}, 54:113--162, 2010.

\bibitem{Newman-99}
M.~Newman and G.~Barkema.
\newblock {\em Monte Carlo Methods in Statistical Physics}.
\newblock Oxford University Press, 1999.

\bibitem{Osborne-10}
M.A. Osborne, R.~Garnett, and S.~Roberts.
\newblock Active data selection for sensor networks with faults and
  changepoints.
\newblock In {\em IEEE International Conference on Advanced Information
  Networking and Applications}, 2010.

\bibitem{Ranzato-10b}
Marc'Aurelio Ranzato, Volodymyr Mnih, and Geoffrey Hinton.
\newblock How to generate realistic images using gated {MRF's}.
\newblock In {\em Advances in Neural Information Processing Systems}, pages
  2002--2010, 2010.

\bibitem{Rasmussen-06}
Carl~Edward Rasmussen and Christopher K~I Williams.
\newblock {\em Gaussian Processes for Machine Learning}.
\newblock MIT Press, Cambridge, Massachusetts, 2006.

\bibitem{Robert-04}
Christian~P. Robert and George Casella.
\newblock {\em Monte {Carlo} Statistical Methods}.
\newblock Springer, 2nd edition, 2004.

\bibitem{Roberts-09}
Gareth~O. Roberts and Jeffrey~S. Rosenthal.
\newblock Examples of adaptive {MCMC}.
\newblock {\em Journal of Computational and Graphical Statistics},
  18(2):349--367, June 2009.

\bibitem{Rosenbluth-55}
Marshall~N. Rosenbluth and Arianna~W. Rosenbluth.
\newblock {M}onte {C}arlo calculation of the average extension of molecular
  chains.
\newblock {\em J. Chem. Phys.}, 23:356--359, 1955.

\bibitem{Rue-09}
Havard Rue, Sara Martino, and Nicolas Chopin.
\newblock Approximate {Bayesian} inference for latent {Gaussian} models by
  using integrated nested {Laplace} approximations.
\newblock {\em Journal Of The Royal Statistical Society Series B},
  71(2):319--392, 2009.

\bibitem{Salakhutdinov-08}
Ruslan Salakhutdinov and Iain Murray.
\newblock On the quantitative analysis of {D}eep {B}elief {N}etworks.
\newblock In {\em International Conference on Machine Learning}, pages
  872--879, 2008.

\bibitem{Santner-03}
T~J Santner, B~Williams, and W~Notz.
\newblock {\em The Design and Analysis of Computer Experiments}.
\newblock Springer, 2003.

\bibitem{Schonlau-98}
Matthias Schonlau, William~J. Welch, and Donald~R. Jones.
\newblock Global versus local search in constrained optimization of computer
  models.
\newblock {\em Lecture Notes-Monograph Series}, 34:11--25, 1998.

\bibitem{Siepmann-92}
J\"{o}rn~I. Siepmann and Daan Frenkel.
\newblock Configurational bias {Monte Carlo}: a new sampling scheme for
  flexible chains.
\newblock {\em Molecular Physics: An International Journal at the Interface
  Between Chemistry and Physics}, 75(1):59--70, 1992.

\bibitem{Smolensky-86}
P.~Smolensky.
\newblock Information processing in dynamical systems: Foundations of harmony
  theory.
\newblock {\em Parallel distributed processing: Explorations in the
  microstructure of cognition, vol. 1: foundations}, pages 194--281, 1986.

\bibitem{Srinivas-10}
Niranjan Srinivas, Andreas Krause, Sham~M Kakade, and Matthias Seeger.
\newblock Gaussian process optimization in the bandit setting: No regret and
  experimental design.
\newblock In {\em International Conference on Machine Learning}, 2010.

\bibitem{Swendsen-87}
Robert~H. Swendsen and Jian-Sheng Wang.
\newblock Nonuniversal critical dynamics in {Monte Carlo} simulations.
\newblock {\em Phys. Rev. Lett.}, 58(2):86--88, 1987.

\bibitem{Swersky-10}
Kevin Swersky, Bo~Chen, Ben Marlin, and Nando de~Freitas.
\newblock A tutorial on stochastic approximation algorithms for training
  restricted {Boltzmann} machines and deep belief nets.
\newblock In {\em Information Theory and Applications Workshop}, pages 1 --10,
  2010.

\bibitem{Tham-02}
S.~S. Tham, A.~Doucet, and R.~Kotagiri.
\newblock Sparse {Bayesian} learning for regression and classification using
  {Markov} chain {Monte Carlo}.
\newblock In {\em International Conference on Machine Learning}, pages
  634--641, 2002.

\bibitem{Vihola-10}
Matti Vihola.
\newblock Grapham: Graphical models with adaptive random walk {Metropolis}
  algorithms.
\newblock {\em Computational Statistics and Data Analysis}, 54(1):49 -- 54,
  2010.

\bibitem{Wang-01}
Fugao Wang and D.~P. Landau.
\newblock Efficient, multiple-range random walk algorithm to calculate the
  density of states.
\newblock {\em Phys. Rev. Lett.}, 86:2050--2053, 2001.

\bibitem{Welsh-99}
Dominic J.~A. Welsh.
\newblock {\em Complexity: knots, colourings and counting}.
\newblock Cambridge University Press, 1993.

\bibitem{Ye-98}
Kenny~Q. Ye.
\newblock Orthogonal column {Latin} hypercubes and their application in
  computer experiments.
\newblock {\em Journal of the American Statistical Association},
  93(444):1430--1439, 1998.

\end{thebibliography}
